\documentclass[a4paper, reqno, 11pt]{article}
\usepackage[utf8]{inputenc}
\usepackage{amsmath}
\usepackage{amsthm}
\usepackage{amssymb}
\usepackage{todonotes}
\usepackage{geometry}
\usepackage{hyperref}
\usepackage{cleveref}
\usepackage[noadjust]{cite}
\newtheorem{theorem}{Theorem}

\newtheorem{lemma}{Lemma}

\newtheorem{proposition}{Proposition}

\newcommand{\id}{\mathbb{I}}
\newcommand{\R}{\mathbb{R}}

\DeclareMathOperator{\Tr}{Tr}

\begin{document}
\title{Kerr metric from two commuting complex structures}
\author{Kirill Krasnov\footnote{Email: kirill.krasnov@nottingham.ac.uk - Corresponding author}  \, and Adam Shaw\footnote{Email: ppyas15@exmail.nottingham.ac.uk} \\ {}\\
{\it School of Mathematical Sciences, University of Nottingham, NG7 2RD, UK}}

\date{August 2024}
\maketitle

\begin{abstract}    The main aim of this paper is to simplify and popularise the construction from the 2013 paper by Apostolov, Calderbank, and Gauduchon, which (among other things) derives the Pleba\'nski-Demia\'nski family of solutions of GR using ideas of complex geometry. The starting point of this construction is the observation that the Euclidean versions of these metrics should have two different commuting complex structures, as well as two commuting Killing vector fields. After some linear algebra, this leads to an ansatz for the metrics, which is half-way to their complete determination. Kerr metric is a special 2-parameter subfamily in this class, which makes these considerations directly relevant to Kerr as well. This results in a derivation of the Kerr metric that is self-contained and elementary, in the sense of being mostly an exercise in linear algebra. 
\end{abstract}

\section{Introduction}

Kerr metric \cite{GravitationalFKerr1963} is a 2-parameter stationary axisymmetric solution to the vacuum Einstein equations, describing a rotating black hole. Its original derivation was by brute force, assuming algebraic speciality and using the null tetrad formalism.\footnote{See e.g. \cite{adamo2016kerrnewmanmetricreview} for a description of the very interesting story behind the Kerr's original derivation.} Standard textbooks on General Relativity just state the metric and proceed to study its properties. Even a check that the Kerr metric satisfies Einstein equations is a rather non-trivial exercise, even thought doable (by hand) with the right formalism, see below. Books that do derive the Kerr metric are rare and few, and the derivation that can be found in \cite{TheMathematicaChandr1998} makes it clear why. The derivation is just too complicated. One can give an "elementary" derivation of the Kerr metric using the the Newman-Janis shift \cite{NoteOnTheKerNewman1965}, see the recent review \cite{adamo2016kerrnewmanmetricreview} for this point of view, but this derivation has a flavour of a magic trick. Another "elementary" derivation is based on the Kerr-Schild form of this metric, see e.g. Chapter 8 of \cite{Deruelle:2018ltn}, or similar in spirit \cite{APossibleIntuNikoli2012}, but there remains some mystery as to why this method should be expected to work. A good further discussion on the issues surrounding the question of simplifying the derivation of the Kerr metric can be found in \cite{PhysicallyMotiBaines2022}.

The Kerr family of solutions is a member of a more general Pleba\'nski-Demia\'nski family \cite{RotatingChargPleban1976}. All metrics of this family can be analytically continued to produce Euclidean signature metrics in 4D. Being an analytic continuation of a Lorentzian algebraically special metric of type D, the resulting Euclidean metrics are type D with respect to both halves of their Weyl curvature. A beautiful result of \cite{SelfDualKahleDerdzi1983} states that a 4D Ricci-flat metric that has the property that one of its Weyl curvature halves is of type D is conformal to a Kahler metric.\footnote{The theorem in \cite{SelfDualKahleDerdzi1983} is more general, and is to be reviewed in the main text.} In particular, there is an integrable almost complex structure. This theorem can be applied with respect to both halves of the Weyl curvature, and tells us that the analytic continuation of any Lorentzian metric that solves the vacuum Einstein equations and is type D is a Riemannian 4D metric that has two distinct integrable almost complex structures. These are of two different orientations, and complex structures in 4D that are of different orientations commute. This means that the Euclidean Kerr metric is expected to have two different commuting complex structures. By assumption, it also has two commuting Killing vector fields. This is a very rigid geometric setup, and it turns out that to a very large extent the Ricci-flat Riemannian metrics of this type can be described completely. This is the viewpoint taken by the work \cite{AmbitoricGeomeAposto2013} that re-derived the Pleba\'nski-Demia\'nski family of metrics as examples of "ambi-toric" 4D geometries, i.e. metrics that are toric with respect to two different complex structures. We remind the reader that a symplectic toric manifold of dimension $2n$ is a symplectic manifold endowed with a hamiltonian action of the torus $(S^1)^n$. 

The main aim of the present paper is to popularise the results of \cite{AmbitoricGeomeAposto2013} to the gravitational physics community, and in particular re-derive the (Euclidean) Kerr metric in a rather elementary fashion. The starting point of the analysis is the (motivated above) assumption of two different commuting complex structures, together with two commuting Killing vector fields. One can then introduce a coordinate system adapted to this rigid geometric setup, as well as the corresponding set of one-forms. This leads to an ansatz for the metric. This ansatz is half-way to the derivation of the Kerr metric, and the metric results by substituting the ansatz into the vacuum Einstein equations, in more or less elementary fashion. We believe that the derivation we present is sufficiently simple that it could be included into a GR textbook. Our aim is also to streamline some of the arguments in \cite{AmbitoricGeomeAposto2013}. This paper is heavy in algebraic geometry, and also its aims are very different from the aim of simplifying the derivation of Kerr. Our aim here is to simplify the arguments in \cite{AmbitoricGeomeAposto2013} as much as possible and derive everything using as elementary manipulations as possible. 

Our other aim is to put to use the Plebanski formalism \cite{Plebanski:1977zz} for 4D GR. In particular, we will see that Derdzi\'nski's theorem (quoted above) is an immediate consequence of the Plebanski formalism. This makes this formalism well-adapted to the problem at hand. So, this article can be read as a rephrasal of the derivation in \cite{AmbitoricGeomeAposto2013} using the Plebanski formalism. 

Finally, our main motivation for undertaking this study was to search for a formalism that would be best adapted to the problem of studying perturbations (in particular gravitational perturbations) of the Kerr metric. We believe the geometry we describe here will prove to be useful for this purpose. More remarks on this are contained in the discussion section. 

This article is organised as follows. We begin, in Section \ref{sec:lit}, with a review of works leading to \cite{AmbitoricGeomeAposto2013}.
In \cref{sec:Chiral_Formalism} we describe the Pleba\'nski formalism. In particular, here we obtain a characterisation of Killing vector fields that is best adapted to this formalism. The material in this section is new. In \cref{sec:ConformalToKahler} we provide a new, Plebanski formalism based proof of the result by Derdzi\'nski that type D metrics are conformal to K\"ahler. We also prove another result by Derdzi\'nski, namely that type D metrics possess a Killing vector field. We review the Kerr solution in \cref{sec:Kerr}, and show by an explicit computation how the two commuting complex structures arise. This section motivates our geometric assumptions about the two Killing vector fields that we need to proceed with the derivation of the metric. In \cref{sec:DerivationOfTypeDSpacetimes} we, more or less, follow \cite{AmbitoricGeomeAposto2013} and derive the Pleba\'nski-Demia\'nski family of metrics. The key point of this derivation is that there exists a natural frame consisting of vector fields that all commute with each other. The dual 1-forms are then all closed. This leads to a natural set of geometric coordinates. In terms of these coordinates, the two complex structures can be characterised explicitly, in terms of just two arbitrary functions, each of one variable. This leads to a remarkably simple ansatz for the metric. To impose Einstein equations, we use the chiral Plebanski formalism. After some analysis, Einstein equations are seen to reduce all unknowns to two quartic polynomials, each of one variable.  
  
  \section{Literature Review}
\label{sec:lit}

Since the discovery of the Kerr metric \cite{GravitationalFKerr1963} in 1963 much effort has been put towards understanding the geometry and properties that it carries.
The Kerr metric is a 2 parameter family of type D solutions to Einstein's equations.
In 1969 Kinnersley parameterised all Lorentzian type D solutions \cite{TypeDVacuumMKinner1969}.
This was not given in a compact form and many different metrics were presented.
Later in 1976 Pleba\'nski and Demia\'nski found a single expression for the most general black hole solution \cite{RotatingChargPleban1976,ANewLookAtTGriffi2005}.
The Pleba\'nski-Demia\'nski (PD) family of metrics is a 7 parameter family of solutions to the Einstein-Maxwell equations.
Removing from the PD metric the electric and magnetic charge parameters, as well as the cosmological constant, we are left with a 4 parameter family of solutions.
All type D spacetimes have Euclidean analogues, obtained by taking an analytic continuation. 
An analytic continuation of the 4-parameter PD subfamily gives a double sided type D Euclidean Ricci-flat metric, which is also asymptotically locally Euclidean (ALE).
More recently, Chen and Teo generalised the Ricci-flat PD metric to an asymptotically locally flat (ALF) metric \cite{AFiveParameteChen2015}, which becomes ALE in the appropriate limit. However, not all of the Chen-Teo metrics have Lorentzian counterparts.
Recently, Biquard and Gauduchon showed \cite{OnToricHermitBiquar2021} that the ALF, Ricci-flat, Hermitian, toric metrics fall into one of the 4 types:
Kerr, Chen-Teo, Taub-NUT and Taub-bolt.

One of the first realisations of the complex properties of type D spacetimes was through the Newman-Janis shift \cite{NoteOnTheKerNewman1965}.
The Newman-Janis shift links the Kerr and Schwarzschild solutions of Einstein's field equations through a complex coordinate transformation.
It was noticed by Flaherty \cite{AnIntegrableSFlaher1974,HermitianAndK1976} that this shift was related to the existence of an integrable almost complex structure that all type D spacetimes are equipped with. Type D spacetimes were thus known to be locally Hermitian.
Later, Derdzi\'nski proved \cite{SelfDualKahleDerdzi1983}, in Euclidean signature, that all one-sided type D metrics (with an extra condition implied by Einstein's equations) are conformal to K\"ahler. Soon after Derdzin\'nski's result, Przanowski and Baka, using this local Hermiticity, reduced the vanishing of the Ricci tensor for a type D metric to a second order PDE for a single function \cite{OneSidedTypePrzano1984}. However, the PDE was written using the complex coordinates making solutions difficult to find.
In 1991 LeBrun derived a related result, with the starting point being a simple form of a K\"ahler metric with a $U(1)$ symmetry \cite{ExplicitSelfDClaude1991}.
The metric is parameterised by 2 functions $u = u(x,y,z)$, $w = w(x,y,z)$ and is given by 
\begin{align}
    g = e^u w(dx^2 + dy^2) + w dz^2 + w^{-1} \alpha^2, \quad \alpha = dt + \theta, \quad \alpha = J(w dz)
\end{align}
where $\theta \in \Lambda^1(\mathbb{R}^3)$ and $J$ is the complex structure. 
The Killing vector is $\partial_t$, which generates the $U(1)$ symmetry.
Vanishing scalar curvature for LeBrun's ansatz becomes the so-called $SU(\infty)$ Toda-Lattice equation
\begin{align}
    u_{xx} + u_{yy} + (e^u)_{zz}=0. \label{eq:SU_Toda_Lattice}
\end{align}
The metric is Ricci flat when $w = c u_z$ for some constant $c$.
This ansatz covers a large class of K\"ahler metrics in dimension 4, however, the Toda equation is difficult to solve in general.
More recently, Tod \cite{OneSidedTypeTodP2020} analysed the case of one-sided type D metrics that have an additional commuting Killing vector field. A more recent related work is \cite{Tod:2024zpa}. In this case the $SU(\infty)$ Toda equation linearises, as shown by Ward in a different context \cite{EinsteinWeylSWard1990}. The solutions can then be obtained from axisymmetric solutions of the flat three-dimensional Laplacian. However, even after the linearisation, deriving a general solution from \cref{eq:SU_Toda_Lattice} remains difficult.
Apostolov, Calderbank and Gauduchon took a different approach in \cite{Aposto2003,AmbitoricGeomeAposto2013}, building upon the fact that many interesting type D metrics are type D with respect to both halves of the Weyl curvature. This is also the approach we follow in the present work.

\section{Chiral Formalism}\label{sec:Chiral_Formalism}

In this section, we will briefly describe a well known and very useful projection of the 4 dimensional Einstein action into one of its chiral halves. This leads to what can be called a chiral formalism for 4D GR. For an introduction to chiral and related formulations of General Relativity see \cite{PlebanskiFormuKrasno2009,FormGenRelGravity2020}.
We will focus on the Euclidean signature metrics, although the formalism generalises to Lorentzian and split signatures.

\subsection{Plebanski action and field equations}

The Hodge star operator on 2-form is an endomorphism, $\star : \Lambda^2 \rightarrow \Lambda^2$, which squares to the identity $\star^2 = \id$.
The Hodge star splits the space of 2-forms in two eigenspaces $\Lambda^2 = \Lambda^+ \oplus \Lambda^-$, with eigenvalues $\pm 1$ respectively.
We call them self-dual (SD) and anti-self-dual 2-forms (ASD) for $+1$ and $-1$ eigenspaces respectively.
Each eigenspace has dimension 3 (as $\dim{\Lambda^2} = 6$), meaning we can prescribe bases $\Sigma^i$ and $\bar{\Sigma}^i$ for $\Lambda^+$ and $\Lambda^-$.
The index $i = 1,2,3$ is an $\mathfrak{so}(3) =\mathfrak{su}(2)$ index, as the Hodge star split is related to the Lie algebra decomposition $\mathfrak{so}(4) = \mathfrak{su}(2) \oplus \mathfrak{su}(2)$. The triple of 2-forms $\Sigma^i$ is called an $SU(2)$ structure, as the choice of such a triple reduces the structure group ${\rm GL}(4,\R)$ to the
 $SU(2)$ preserving each of the 2-forms. 
 
 The Riemann curvature can be viewed as another endomorphism acting on the space of 2-forms. It can then be decomposed with respect to the decomposition of the space of 2-forms. As is well-known, the Riemann tensor takes the following form
\begin{align}
    R_{\mu \nu \rho \sigma} = \begin{pmatrix}
        W^+ + R & \mathcal{R} \\
        \mathcal{R} & W^- + R
\end{pmatrix} \label{eq:Riemann_SD_ASD_Decomposition}
\end{align}
where $W^\pm$ are the SD/ASD halves of the Weyl curvature tensor, $\mathcal{R}$ is the traceless part of the Ricci tensor and $R$ is the remaining trace of Ricci.
The first row of \cref{eq:Riemann_SD_ASD_Decomposition} is the self-dual projection of the Riemann curvature with respect to one of the pairs of its indices. It is now easy to see that the Einstein condition is equivalent to requiring the ASD part of the first row in \cref{eq:Riemann_SD_ASD_Decomposition} to be zero. 
This is useful because one only needs half the of curvature components to enforce the Einstein condition. The action that realises this idea and has 
Einstein metrics at its critical points is the Pleba\'nski action \cite{Plebanski:1977zz} 
\begin{align}
    S = \int \Sigma^i \wedge F^i - \frac{1}{2} \left( \Psi^{ij} + \frac{\Lambda}{3} \delta^{ij} \right) \Sigma^i \wedge \Sigma^j. \label{eq:Plebanski_Action}
\end{align}
Here $F^i = dA^i + \frac{1}{2}\epsilon^{ijk}A^j \wedge A^k$ is the curvature of the $\mathfrak{su}(2)$ connection $A^i \in \Lambda^1\otimes\mathfrak{su}(2)$.
The Lagrange multiplier $\Psi^{ij}$ is required to be symmetric and traceless, $\Lambda$ is the cosmological constant.
Varying with respect to $\Psi^{ij}, A^i$ and $\Sigma^i$ one obtains the following Euler-Lagrange equations
\begin{gather}
    \Sigma^i \wedge \Sigma^j \sim \delta^{ij} \nonumber\\
    d^A \Sigma^i = d\Sigma^i + \epsilon^{ijk} A^j \wedge \Sigma^k = 0 \label{eq:SD_Pleb_Field_Equations}\\
    F^i = \left( \Psi^{ij} + \frac{\Lambda}{3}\delta^{ij} \right) \Sigma^j. \nonumber
\end{gather}
The first of the above equations implies that the 2-forms describe only the 10 components of the metric plus the 3 components defining an $SO(3)$ frame.
It is useful to know that given a tetrad basis $e^I = e^0, e^i$ one can construct the self-dual and anti-self-dual 2-forms like so
\begin{align}
    \Sigma^i &= e^0 \wedge e^i - \frac{1}{2}\epsilon^{ijk}e^j \wedge e^k \label{eq:SelfDualTwoForms_From_Tetrad} \\
    \bar{\Sigma}^i &= e^0 \wedge e^i + \frac{1}{2} \epsilon^{ijk} e^j \wedge e^k. \label{eq:AntiSelfDualTwoForm_From_Tetrad}
\end{align}
The components of the 2-forms satisfy the quaternion algebra
\begin{align}
    \Sigma^i = \frac{1}{2}\Sigma^i_{\mu \nu} dx^\mu \wedge dx^\nu, \quad \Sigma^i_\mu{}^\alpha \Sigma^j_\alpha{}^\nu = -\delta^{ij}\delta_\mu^\nu + \epsilon^{ijk} \Sigma^k_\mu{}^\nu. \label{eq:Sigma_Quaternion_AlgebraSigma_Quaternion_Algebra}
\end{align}
The second equation in (\ref{eq:SD_Pleb_Field_Equations}) implies that the 1-form $A^i$ corresponds to the self-dual part of the Levi-Civita connection for the metric defined by the $\Sigma^i$'s, see \cite{PlebanskiFormuKrasno2009,FormGenRelGravity2020}, and also a more recent reference \cite{Bhoja:2024xbe}.

To analyse the final equation we look at the general decomposition of the SD part of the curvature 
\begin{align}
    F^i = \left( \Psi^{ij} + \frac{R}{3} \delta^{ij} \right) \Sigma^j + R^{ij} \bar{\Sigma}^j
\end{align}
The last line in \cref{eq:SD_Pleb_Field_Equations} then imposes that $R = \Lambda$ and $R^{ij} = 0$, which are precisely (vacuum) Einstein's equations.
This discussion shows that the critical points of the Plebanski action are ${\rm SU}(2)$ structures that are Einstein, in the sense that the metric defined by the triple of 2-forms according to the formula
\begin{align}\label{Urb-metric}
    g(u,v) \nu_g = \frac{1}{6}\epsilon^{ijk} \iota_u \Sigma^i \wedge \iota_v \Sigma^j \wedge \Sigma^k
\end{align}
is Einstein. Here $\nu_g$ is the volume form for the metric $g$, and both sides of the relation are top forms.

\subsection{SU(2) Structures and Killing Vectors}\label{sec:SU2_Killing_Vectors}

Killing vectors fields are infinitesimal symmetries of the metric $\mathcal{L}_X g_{\mu \nu} = 0$, where $\mathcal{L}$ is the Lie derivative. For our purposes below we need to characterise Killing vectors from the point of view of $SU(2)$ structures. Given that the metric (\ref{Urb-metric}) defined by a triple of 2-forms $\Sigma^i$ is invariant with respect to the ${\rm SO}(3)$ action on $\Sigma^i$, it is suggestive that in order for a vector field $X$ to be infinitesimal isometries at the level of $\Sigma^i$, the Lie derivative of $\Sigma^i$ should be an ${\rm SO}(3)$ transformation
 \begin{align}\label{def:KillingVectorDef}
        \mathcal{L}_X \Sigma^i = \epsilon^{ijk} \tilde{\theta}^j \Sigma^k.
    \end{align}
We will call a vector field $X$ satisfying this property a {\bf $\Sigma$-Killing vector field}. 

The following lemma describes an equivalent way of stating the $\Sigma$-Killing condition
\begin{lemma}
    $\Sigma$-Killing vectors $X \in TM$ can be equivalently described as those satisfying 
    \begin{align}
       d^A \iota_X \Sigma^i = \epsilon^{ijk} \theta^j \Sigma^k,
    \end{align}
    where $d^A$ is the exterior covariant derivative with respect to the $\Sigma$-compatible $d^A \Sigma^i=0$ connection $A^i$. 
    \end{lemma}
\begin{proof}
    To see this we use the Cartan's magic formula
    \begin{align}
        \mathcal{L}_X \Sigma^i = \iota_X d\Sigma^i + d \iota_X \Sigma^i.\label{eq:Cartans_Magic_Formula}
    \end{align}
   We now require this to be a gauge transformation    
       \begin{align}
        \mathcal{L}_X \Sigma^i = \epsilon^{ijk} \tilde{\theta}^j \Sigma^k \label{eq:SU2_LieD_KillingEqn}.
    \end{align}
    Using $d^A \Sigma^i = 0$ and \cref{eq:Cartans_Magic_Formula} we find 
    \begin{align}
        \mathcal{L}_X \Sigma^i = d^A \iota_X \Sigma^i + \epsilon^{ijk} (\iota_X A)^j \Sigma^k = \epsilon^{ijk} \tilde{\theta}^j \Sigma^k  \quad \Rightarrow \quad d^A \iota_X \Sigma^i = \epsilon^{ijk} \theta^j \Sigma^k
    \end{align}
    where $\theta^i = \tilde{\theta}^i - \iota_X A^i$ is the shifted $SO(3)$ gauge transformation parameter.
\end{proof}

Next we now link $\Sigma$-Killing vector fields with the usual Killing vector fields.
\begin{proposition}
    A $\Sigma$-Killing vector field is a Killing vector field for the metric (\ref{Urb-metric}) defined by the triple of 2-forms $\Sigma^i$. 
\end{proposition}

Before we prove this statement, we need to remind the reader some facts about the decomposition of $\mathfrak{su}(2)$-valued 2-forms. This construction is from \cite{Bhoja:2024xbe}, and we will only summarise the main points. Let us denote $E=\mathfrak{su}(2)$, and consider the space of $E$-valued 2-forms. We introduce an endomorphism acting on this space 
\begin{align}
    J_2 : \Lambda^2 \otimes E \to \Lambda^2 \otimes E, \quad J_2(B)^i_{\mu\nu} = \epsilon^{ijk} \Sigma^j_{[\mu}{}^\alpha B^k_{|\alpha|\nu]}, \quad B^i_{\mu \nu} \in \Lambda^2 \otimes E.    
\end{align}
The reference \cite{Bhoja:2024xbe} shows that this operator satisfies 
\begin{align}
    J_2 (J_2 - 2)(J_2 - 1)(J_2 + 1) = 0,
\end{align}
and thus has eigenvalues $0,-1,1,2$. The space $\Lambda^2 \otimes E$ decomposes into the eigenspaces of this operator
\begin{align}
    \Lambda^2 \otimes E = (\Lambda^2 \otimes E)_1 \oplus (\Lambda^2 \otimes E)_3 \oplus (\Lambda^2 \otimes E)_5 \oplus (\Lambda^2 \otimes E)_9.
\end{align}
The subscript in $(\Lambda^2 \otimes E)_k$ indicates the dimension of the corresponding space. The spaces $(\Lambda^2 \otimes E)_{1,3,5,9}$ are the eigenspaces of $J_2$ of eigenvalue $2,1,-1,0$ respectively, see section 4.4 of \cite{Bhoja:2024xbe} for more details. It can also be shown that $(\Lambda^2 \otimes E)_9 = \Lambda^- \otimes E$. Here $\Lambda^-$ is the space of anti-self dual 2-forms, with $\Lambda^2=\Lambda^+\oplus\Lambda^-$. We also have the following explicit projections to each of the eigenspaces
\begin{align}\nonumber
B^{(i}_{\alpha\beta} \Sigma^{j)\alpha\beta} \in S^2(E), \\ \label{irreducibles}
\epsilon^{ijk} B^j_{\alpha\beta} \Sigma^{k\alpha\beta} \in E, \\ \nonumber
B^i_{(\mu|\alpha|} \Sigma^{i\alpha}{}_{\nu)} \in S^2(T^*M) .
\end{align}
The described maps project away some of the irreducible components, while keeping the others. Thus, the kernel of the first map is $(\Lambda^2 \otimes E)_{3+9}$, and so this is a projection onto the $(\Lambda^2 \otimes E)_{5+1}$ component. The second map is a projection onto the $(\Lambda^2 \otimes E)_3$ component. The last is a projection onto the 
$(\Lambda^2 \otimes E)_{1+9}$ component. The $(\Lambda^2 \otimes E)_1$ component, which contains multiples of $\Sigma^i$, is contained in both the first and last projection as the trace part. 

\begin{proof}
We now prove the proposition. To this end, we consider the $E$-valued 2-form $d^A \iota_X \Sigma^i$, for an arbitrary vector field $X$, and compute its irreducible parts. We have
\begin{align}
    (d^A \iota_X \Sigma^i)_{\mu \nu} = \partial_{[\mu} ( X^{\rho} \Sigma^i_{|\rho| \nu]} )  + \epsilon^{ijk} A^j_{[\mu} X^{\rho} \Sigma^i_{|\rho| \nu]}.
   \end{align}
We can replace the partial derivative in the first term by the covariant derivative with respect to the metric, because it is applied to a 1-form, and the result is $\mu\nu$ anti-symmetrised. We thus have    
  \begin{align}\label{d-i-Sigma}
    (d^A \iota_X \Sigma^i)_{\mu \nu}   = \nabla_{[\mu} ( X^{\rho} \Sigma^i_{|\rho| \nu]} )  + \epsilon^{ijk} A^j_{[\mu} X^{\rho} \Sigma^i_{|\rho| \nu]}.
     \end{align} 
 We now use the fact that the connection 1-forms $A^i_\mu$ are the "intrinsic torsion" components of the ${\rm SU}(2)$ structure, i.e. we have the following equation
  \begin{align}\label{torsion}
 \nabla_\mu \Sigma^i_{\rho\nu} + \epsilon^{ijk} A^j_\mu \Sigma^k_{\rho\nu}=0.
  \end{align} 
 This fact is proven in  \cite{Bhoja:2024xbe}, by considering the $E$-valued 2-forms $X^\mu \nabla_\mu \Sigma^i_{\rho\sigma}$, and showing that only the $(\Lambda^2 \otimes E)_3$ projection is non-vanishing. This is then parametrised by $X^\mu A_\mu^i$. Contracting (\ref{torsion}) with $X^\rho$ and $\mu\nu$ anti-symmetrising we see that we can rewrite  (\ref{d-i-Sigma}) as
  \begin{align}   
    (d^A \iota_X \Sigma^i)_{\mu \nu}  = \Sigma^i_{[\mu}{}^\rho \nabla_{\nu]} X_\rho.
\end{align}
We can now use the algebra \cref{eq:Sigma_Quaternion_AlgebraSigma_Quaternion_Algebra} of $\Sigma$'s to compute the projections of 
 $d^A \iota_X \Sigma^i$ onto the irreducible components. We have
\begin{align}
    (\Lambda^2 \otimes E)_1&:  \quad \Sigma^{i\mu\nu} (d^A \iota_X \Sigma^i)_{\mu \nu} = 3 \nabla_\mu X^\mu \nonumber\\
    (\Lambda^2 \otimes E)_3&: \quad \epsilon^{ijk} \Sigma^{j\mu \nu} (d^A \iota_X \Sigma^k)_{\mu \nu} = -2\Sigma^{i \mu \nu} \nabla_{[\mu} X_{\nu]} \\
    (\Lambda^2 \otimes E)_5&: \quad \Sigma^{\langle i | \mu \nu} (d^A \iota_X \Sigma^{|j\rangle})_{\mu \nu} = 0 \nonumber\\
    (\Lambda^2 \otimes E)_9&: \quad \Sigma^i_{\langle \mu|}{}^\alpha (d^A \iota_X \Sigma^i)_{\alpha | \nu \rangle} = -3 \nabla_{\langle \mu} X_{\nu \rangle} \nonumber
\end{align}
Here $\langle \mu \nu \rangle$ denotes the symmetric traceless part.
We now assume that the $\Sigma$-Killing equation is satisfied and so $d^A \iota_X \Sigma^i= \epsilon^{ijk} \theta^j \Sigma^k$. The right-hand side here has the only non-vanishing projection 
\begin{align}
    (\Lambda^2 \otimes E)_3: \quad \epsilon^{ijk} \Sigma^{j\mu\nu} (\epsilon^{klm} \theta^l \Sigma^m)_{\mu\nu} = 8 \theta^i.
\end{align}
It is then clear that the $\Sigma$-Killing equation is equivalent to 
\begin{align}
    \nabla_{(\mu} X_{\nu)} = 0, \quad \theta^i = -\frac{1}{4}\Sigma^{i\mu\nu} \nabla_{[\mu} X_{\nu]}.
\end{align}
The first is the usual Killing equation, proving the proposition. The last relation gives $\theta^i$ in terms of the vector field $X$. In words, $\theta^i$ is the vector giving the self-dual projection of $dX^\sharp$, where $X^\sharp$ is the one-form corresponding to the vector field $X$. 
\end{proof}

\section{One-sided type D metrics are conformal to K\"ahler} \label{sec:ConformalToKahler}

The aim of this section is to present an alternative (and new) proof of the theorem by Derdzi\'nski \cite{SelfDualKahleDerdzi1983}. We use the chiral formalism, which makes the statement almost manifest. We also prove another statement from \cite{SelfDualKahleDerdzi1983}, namely that all one-sided type D Einstein manifolds have a Killing vector. 

\subsection{Derdzi\'nski theorem}

We start by using the chiral formalism to prove the following theorem 
\begin{theorem}
    Let $(M,g)$ be a smooth 4-manifold that is one-sided type D, i.e. with one of the two halves of the Weyl curvature, say SD, having two coinciding eigenvalues. Assume in addition that the SD part of the Weyl curvature is divergence-free $\nabla \cdot W^+ = 0$. Then $g$ is conformal to a K\"ahler metric. \label{thrm:DerdziConfKahler}
\end{theorem}

This was first proved by Derdzi\'nski in \cite{SelfDualKahleDerdzi1983}. Note that there is no assumption here that the metric is Einstein, one only assumes that the SD Weyl tensor is divergence-free. This is true for Einstein manifolds (as we will also see later), but is a weaker condition. The proof we present here is an elementary consequence of the chiral Plebanski formalism. In fact, one can motivate the usefulness of the chiral formalism by the fact that this, reasonably non-trivial statement of Riemannian geometry, becomes apparent in this formalism. 

We start by translating the devergence-free condition as a condition on objects in the Plebanski formalism. As is clear from \cref{eq:Riemann_SD_ASD_Decomposition}, the self-dual part of the Weyl curvature is directly related to the matrix $\Psi^{ij}$, we have $W^+_{\mu \nu \rho \sigma} = \Psi^{ij} \Sigma^i_{\mu\nu} \Sigma^j_{\rho\sigma}$.
The condition that the self-dual Weyl curvature is divergence-free translates to
\begin{align}\label{div-W}
    0 = (\nabla \cdot W^+)_{\nu \rho \sigma} = \nabla^\mu W^+_{\mu \nu \rho \sigma} = \nabla^{A\mu} \left( \Psi^{ij} \Sigma^i_{\mu \nu} \Sigma^j_{\rho\sigma} \right) = \nabla^{A\mu} \Psi^{ij} \Sigma^i_{\mu \nu} \Sigma^j_{\rho \sigma},
\end{align}
where $\nabla^A_\mu$ is the total gauge covariant derivative $\nabla^A_\mu \theta^i_\nu = \nabla_\mu \theta^i_\nu + \epsilon^{ijk} A^j_\mu \theta^k_\nu$ and $A^i$ is the self-dual connection defined by $d^A \Sigma^i = 0$. As we have already remarked in (\ref{torsion}), by construction, this covariant derivative has the property $\nabla_\mu^A \Sigma^i_{\rho\sigma}=0$, which justifies the last equality. 

The object on the right-hand side of (\ref{div-W}) is self-dual with respect to the index pair $\rho\sigma$. It can therefore be converted to an $E$-valued one-form 
\begin{align}\label{div-W-1} 
    (\nabla \cdot W^+)_{\nu \rho \sigma} \Sigma^{i\rho\sigma} = 4\nabla^{A\mu}(\Psi^{ij} \Sigma^j_{\mu \nu}) =4 \nabla^{A\mu}(\Psi^{ij}) \Sigma^j_{\mu \nu}.
\end{align}
This can be rewritten in form notations. Considering the 3-form $(\nabla^A \Psi^{ij}) \Sigma^j$, and taking the Hodge star we obtain a multiple of (\ref{div-W-1}). Thus, we see that the following two statements are equivalent
\begin{align}
    (\nabla \cdot W^+) = 0 \quad \Longleftrightarrow   \quad d^A (\Psi^{ij} \Sigma^j) = 0. \label{eq:Divergence_Free_SD_Weyl}
\end{align}

We now have the following simple corollary of the Plebanski formalism
\begin{proposition}
    Every Einstein 4-Manifold has divergence free self-dual Weyl curvature, $\nabla \cdot W^+ = 0$. \label{rem:Einstein_Divergence_Free}
\end{proposition}
\begin{proof}[Proof of \cref{rem:Einstein_Divergence_Free}]
    The proof follows from the Bianchi identity $d^A F^i= 0$. In Plebanski formalism, the Einstein equation takes the following form
    \begin{align}\label{Einstein}
        F^i = \left( \Psi^{ij} + \frac{\Lambda}{3} \delta^{ij} \right)\Sigma^j.
    \end{align}
    Applying the exterior covariant derivative on the left, and using the Bianchi identity, we have
    \begin{align}
        0 = d^A F^i = d^A \left( \Psi^{ij} \Sigma^j + \frac{\Lambda}{3} \Sigma^i \right) = d^A (\Psi^{ij} \Sigma^j ) = 0
    \end{align}
    where we have used $d^A \Sigma^i = 0$ and that $\Lambda$ is constant.
    Using \cref{eq:Divergence_Free_SD_Weyl} we see that Einstein 4-manifolds have divergence free self-dual Weyl curvature. 
\end{proof}

\begin{proof}[Proof of \cref{thrm:DerdziConfKahler}]
    Let us denote the eigenvalues of $\Psi^{ij}$ as $\alpha,\beta,\gamma$. We can assume, without loss of generality that $\alpha \geq \beta \geq \gamma$. 
    The traceless condition on $\Psi^{ij}$ imposes $\alpha + \beta + \gamma = 0$.
    This means that $\alpha > 0$ and $\gamma < 0$.
    There are then two possibilities for two eigenvalues to coincide. One is when the middle eigenvalue coincides with the negative one $\beta = \gamma$. The other is when the two positive eigenvalues coincide. Thus, the matrix $\Psi^{ij}$ is of the two possible forms
    \begin{align}
        \Psi^{ij} = \begin{pmatrix}
            2\alpha & 0 & 0 \\
            0 & -\alpha & 0 \\
            0 & 0 & -\alpha
        \end{pmatrix} \label{eq:TypeD_Weyl_Curvature_Matrix}, \qquad 
         \Psi^{ij} = \begin{pmatrix}
            \alpha & 0 & 0 \\
            0 & \alpha & 0 \\
            0 & 0 & -2\alpha
        \end{pmatrix}
    \end{align}
    with $\alpha$ being a positive real function. The two cases differ by the sign of ${\rm det}(\Psi)$. See section 2 of \cite{Biquard:2023gyl} for a related discussion. We consider the first case in details. The second case is treated analogously, by observing that it can be mapped to the first by allowing $\alpha$ to be negative and relabelling the basis vectors. In the proof below, we never need to assume that $\alpha>0$, and so the proof covers both of the cases. 
    
    We conclude that for every one-sided type D Riemannian metric, one can always pick a basis $\Sigma^i$ in the space of SD 2-forms, so that $\Psi^{ij}$ is diagonal and of the form of the first matrix in \cref{eq:TypeD_Weyl_Curvature_Matrix}. With this form of $\Psi^{ij}$ the equations (\ref{eq:Divergence_Free_SD_Weyl}) become
    \begin{align}
        d\left( 2\alpha \Sigma^1 \right) + A^2 \left( -\alpha \Sigma^3 \right) - A^3 \left( -\alpha \Sigma^2 \right) &= 0 \nonumber\\
        d\left( -\alpha \Sigma^2 \right) + A^3 \left( 2\alpha \Sigma^1 \right) - A^1 \left( -\alpha \Sigma^3 \right) &= 0 \label{eq:Explicit_Divergence_Equations}\\
        d\left( -\alpha \Sigma^3 \right) + A^1 \left( -\alpha \Sigma^2 \right) - A^2 \left( 2\alpha \Sigma^1 \right) &= 0. \nonumber
    \end{align}
    Using also $d^A \Sigma^i = 0$, these equations can be rewritten as follows 
    \begin{gather}
        2 d\alpha \Sigma^1 + 3\alpha d\Sigma^1 = 0\\
        d\alpha \Sigma^2 + 3\alpha (d\Sigma^2 - A^1 \Sigma^3) = 0, \quad d\alpha \Sigma^3 + 3\alpha (d\Sigma^3 + A^1 \Sigma^2) = 0. \label{eq:dSigma_Equations_123}
    \end{gather}
    The first equation in the above can be written as 
    \begin{align}
        d\omega = 0, \quad \omega:= \alpha^\frac{2}{3} \Sigma^1 \label{eq:ConfKahler_Closed_2form_Equation}.
    \end{align}
    We now introduce a scaled complex linear combinations of the remaining 2-forms
    \begin{align}
        \Omega := \alpha^\frac{2}{3} (\Sigma^2 + i \Sigma^3) \quad (resp. \quad \bar{\Omega} := \alpha^{\frac{2}{3}} (\Sigma^2 - i\Sigma^3)),
    \end{align}
    where the scaling factor in front is the same as for $\omega$. The last two equations in \cref{eq:Explicit_Divergence_Equations} can be rewritten as
    \begin{align}\label{d-Sigma-23}
    d\alpha (\Sigma^2+i \Sigma^3) + 3\alpha d(\Sigma^2+i \Sigma^3) + 3i \alpha A^1 (\Sigma^2+i \Sigma^3)=0.
    \end{align}
    For later purposes, we note that using $d^A \Sigma^i=0$ we can also rewrite the last two equations as
    \begin{align}
    d\alpha \Sigma^2 - 3\alpha A^3\Sigma^1=0, \qquad  d\alpha \Sigma^3 + 3\alpha A^2\Sigma^1=0,
   \end{align}
    and thus their complex linear combination gives
    \begin{align}\label{d-Omega-Sigma-1}
    d\alpha (\Sigma^2+i \Sigma^3) + 3i \alpha (A^2+i A^3) \Sigma^1=0.
    \end{align}
    This will be useful later. 
    
  We now use (\ref{d-Sigma-23}) to get an equation for $\Omega$
    \begin{align}\label{d-Omega}
        d\Omega +\left( i A^1 - \frac{d\alpha}{3\alpha}\right) \Omega = 0.
        \end{align}
We will rewrite this equation in a more suggestive way below.    
Finally, the algebraic equations $\Sigma^i\Sigma^j\sim \delta^{ij}$ imply the following equations
    \begin{align}
        \omega \Omega = 0, \quad \Omega \Omega = 0, \quad \omega^2 = \frac{1}{2}\Omega \bar{\Omega} \label{eq:ConfKahler_metricity}.
    \end{align}
    
    We can now understand where the conformal to K\"ahler comes from. The second equation in (\ref{eq:ConfKahler_metricity}) implies that the complex 2-form $\Omega$ is decomposable. Its two 1-form factors can be declared to be the $(1,0)$ forms of an almost complex structure, which defines this almost complex structure. The equation in (\ref{d-Omega}) can then be rewritten in a better way. It is clear that it states
    \begin{align}
        d\Omega +\left( i A^1 - \frac{1}{3\alpha} d\alpha \Big|_{(0,1)} \right) \Omega = 0,
        \end{align}
        because only the projection of $d\alpha$ to the space $\Lambda^{0,1}$ survives the wedge product with $\Omega$. We can then add another term that is valued in $\Lambda^{1,0}$ and rewrite this equation as 
          \begin{align}
        d\Omega +i a \Omega = 0, \qquad a = A^1 + \frac{1}{3i \alpha} ( d\alpha \Big|_{(1,0)}  - d\alpha \Big|_{(0,1)} ).\label{eq:ConfKahler_IntegrableComplexStrcture}
        \end{align}
        The connection 1-form $a$ is now real-valued. We note that it can also be written as       
        \begin{align}
        \quad a_\mu = A^1_\mu + \omega_\mu{}^\nu \partial_\nu \ln\left(\alpha^\frac{1}{3}\right). 
    \end{align}
    
    The equation (\ref{eq:ConfKahler_IntegrableComplexStrcture}) implies that this complex structure is integrable. Indeed, it implies that $d\Omega$ has only a $\Lambda^{2,1}$ part, which is equivalent to integrability. The first of the equations in (\ref{eq:ConfKahler_metricity}) implies that $\omega$ is in $\Lambda^{1,1}$ with respect to this complex structure. It can thus be declared to be the K\"ahler form. The last equation in  (\ref{eq:ConfKahler_metricity}) is the correct normalisation condition linking the volume form as defined by the K\"ahler form $\omega$ and the volume form defined by $\Omega$. The equation (\ref{eq:ConfKahler_Closed_2form_Equation}) states that the K\"ahler form is closed. Altogether, we deduce that the metric defined by $\omega,\Omega$ is K\"ahler, with Ricci form $\rho = da$. It is also clear that the original metric, defined by the $\Sigma^i$'s, is conformal to this K\"ahler metric 
    \begin{align}
        g_K = \alpha^\frac{2}{3} g_\Sigma
    \end{align}
    This proves \cref{thrm:DerdziConfKahler}. 
\end{proof}

%

\subsection{Killing Vectors in Einstein one-sided type D spaces}

Another statement proven in \cite{SelfDualKahleDerdzi1983} is about the existence of Killing vectors in every Einstein manifold satisfying the assumptions of \cref{thrm:DerdziConfKahler}. 
\begin{proposition} \label{prop:ConfKahlerKillingVector}
    Let $(M,\Sigma^i)$ be an Einstein manifold with type D self-dual Weyl tensor. Let $\Sigma^i$ be a basis for the triple of self-dual 2-forms chosen so that the SD part of the Weyl curvature is diagonal, and $\Sigma^1$ is in the direction of the special (non-repeated) eigenvalue. Then $X^\mu = \Sigma^{1\mu \nu} \nabla_\nu \alpha^{-\frac{1}{3}}$ is a Killing vector field. Here $\alpha$ with $\alpha^2 = \frac{1}{6}\Tr{\Psi^2}$ is the repeated eigenvalue of the self-dual Weyl tensor.
\end{proposition}
\begin{proof}
    Inserting $X$ into $\Sigma^1$ we find
    \begin{align}
        (\iota_X \Sigma^1)_\mu = X^\nu \Sigma^1_{\nu \mu} = \Sigma^{1 \nu \rho} \nabla_\rho \alpha^{-\frac{1}{3}} \Sigma^1_{\nu \mu} = \nabla_\mu \alpha^{-\frac{1}{3}} \quad \Rightarrow \quad \iota_X \Sigma^1 = d\alpha^{-\frac{1}{3}}
    \end{align}
    where we have used the quaternion algebra of the 2-forms. The other two 1-forms are computed as
    \begin{align}
    (\iota_X \Sigma^2)_\mu = \Sigma^{1\sigma\rho} \nabla_\rho \alpha^{-\frac{1}{3}} \Sigma^2_{\sigma\mu}= \Sigma^1_\mu{}^\sigma \Sigma^2_\sigma{}^\rho \nabla_\rho \alpha^{-\frac{1}{3}},
      \end{align}
      where we have used the algebra of $\Sigma$'s to anti-commute $\Sigma^{1,2}$. We similarly have
     \begin{align}
     (\iota_X \Sigma^3)_\mu = \Sigma^1_\mu{}^\sigma \Sigma^3_\sigma{}^\rho \nabla_\rho \alpha^{-\frac{1}{3}}.
     \end{align}
     Introducing $\Sigma^+ = \Sigma^2+ i \Sigma^3$ we have
      \begin{align}\label{i-X-Sigma-plus}
      (\iota_X \Sigma^+)_\mu = \Sigma^1_\mu{}^\sigma \Sigma^+_\sigma{}^\rho \nabla_\rho \alpha^{-\frac{1}{3}}.
      \end{align}
      We now use (\ref{d-Omega-Sigma-1}), which we rewrite as 
      \begin{align}
      d\alpha^{-1/3} \Sigma^+ = i \alpha^{-1/3} A^+ \Sigma^1,
      \end{align}
      where we introduced $A^+=A^2+ iA^3$. Taking the Hodge dual we get
       \begin{align}
      \Sigma^+_\mu{}^\rho \nabla_\rho \alpha^{-1/3}  = i \alpha^{-1/3} \Sigma^1_\mu{}^\rho A^+_\rho.
      \end{align}
      Using this in (\ref{i-X-Sigma-plus}) we get
       \begin{align}
      (\iota_X \Sigma^+)_\mu =  i \alpha^{-1/3} \Sigma^1_\mu{}^\sigma  \Sigma^1_\sigma{}^\rho A^+_\rho = -i \alpha^{-1/3}A^+_\mu.
      \end{align}
      We now summarise the above results
      \begin{align}\label{i-X-Sigma-plus*}
      \iota_X \Sigma^1 = d \alpha^{-1/3}, \qquad \iota_X \Sigma^+= -i \alpha^{-1/3}A^+, \qquad  \qquad \iota_X \Sigma^-= i \alpha^{-1/3}A^-.
 \end{align}
 
  Having computed the components of $\iota_X \Sigma^i$, we can compute its exterior covariant derivative and use the characterisation of the $\Sigma$-Killing vectors to show that $X$ is a Killing vector field. It will be convenient to rewrite the formulas for the exterior covariant derivative of $\iota_X \Sigma^i$ in terms of the introduced objects $\Sigma^\pm, A^\pm$. We have
   \begin{align}
  d_A \iota_X \Sigma^1 = d \iota_X \Sigma^1 + \frac{1}{2i} (A^- \iota_X \Sigma^+ - A^+ \iota_X \Sigma^-), \\ \nonumber
  d_A \iota_X\Sigma^+ = d  \iota_X \Sigma^+ - i A^+ \iota_X \Sigma^1 + i A^1 \iota_X \Sigma^+.
    \end{align}
    Substituting (\ref{i-X-Sigma-plus*}) into the first line we see it vanishes $d_A \iota_X \Sigma^1=0$. For the first term this is manifest from $d^2=0$. For the second term there is a cancellation of the two $A^+ A^-$ terms. For the second line we have
      \begin{align}
     d_A \iota_X\Sigma^+ = -i d( \alpha^{-1/3} A^+)  - i A^+ d \alpha^{-1/3} +  A^1 \alpha^{-1/3}A^+ = - i \alpha^{-1/3} (dA^+ +i A^1 A^+).
      \end{align}
What arises here is the curvature 
 \begin{align}
 F^+ = F^2 +i  F^3 = dA^+ +i A^1 A^+.
 \end{align}
 Using the Einstein equations in the form (\ref{Einstein}) we have
 \begin{align}
 F^+ = (-\alpha + \frac{\Lambda}{3}) \Sigma^+.
\end{align}
 This means that all in all
 \begin{align}
 d_A \iota_X\Sigma^1=0, \qquad d_A \iota_X\Sigma^+ = -i \alpha^{-1/3}(-\alpha + \frac{\Lambda}{3}) \Sigma^+.
 \end{align}
    
    On the other hand, the $\Sigma$-Killing condition with the vector $\theta^i=(\theta,0,0)$ becomes
    \begin{align}
    d_A \iota_X \Sigma^1 = 0, \qquad d_A \iota_X\Sigma^+= i \theta \Sigma^+.
     \end{align}
     We see that this is indeed satisfied, with
     \begin{align}
     \theta = \alpha^{-1/3}(\alpha - \frac{\Lambda}{3}).
     \end{align}
     This finishes the proof. 
   \end{proof}
   
   \section{Two complex structures of Euclidean Kerr}
   \label{sec:Kerr}
   
The aim of this section is to exhibit the structures described in the previous section for the example of the Euclidean Kerr metric. This section serves only to motivate the considerations that follow. 

\subsection{Metric and Chiral objects}

The Euclidean Kerr metric is given by
\begin{align}
    g = \left( 1- \frac{2Mr}{\rho^2} \right) dt^2 + \frac{\rho^2}{\Delta} dr^2 + \rho^2 d\theta^2 + \left( r^2 -a^2 - \frac{2Mr a^2}{\rho^2} \sin(\theta)^2 \right) \sin(\theta)^2 d\phi^2 + \frac{4Mra\sin(\theta)^2}{\rho^2} dt d\phi,
    \label{eq:EuclideanKerrMetric}
\end{align}
where 
\begin{align}\label{rho-delta}
\rho^2 = r^2 - a^2 \cos(\theta)^2, \qquad \Delta = r^2 - 2Mr - a^2.
\end{align}
This is obtained from the usual Lorentzian Kerr metric in Boyer-Lindquist coordinates through a transformation $t \rightarrow it$, $a \rightarrow ia$.
A useful tetrad basis is given by
\begin{align}
    e^0 = \frac{\sqrt{\Delta}}{\rho}(dt - a \sin(\theta)^2 d\phi), \quad e^1 = \frac{\rho}{\sqrt{\Delta}}dr,\quad e^2 = \rho d\theta,\quad e^3 = \frac{\sin(\theta)}{\rho}(a dt + (r^2 - a^2) d\phi).\label{eq:KerrTetradBasis}
\end{align}

\subsection{Self-dual objects}

The self-dual 2-forms are readily calculated from the \cref{eq:SelfDualTwoForms_From_Tetrad} using the tetrad basis.
We obtain 
\begin{align}
    \Sigma^1 &= (dt-a\sin(\theta)^2 d\phi) \wedge dr - \sin(\theta) d\theta \wedge (a dt + (r^2-a^2) d\phi)\\
    \Sigma^2 &= \sqrt{\Delta}(dt-a\sin(\theta)^2 d\phi) \wedge d\theta -\frac{\sin(\theta)}{\sqrt{\Delta}}(a dt + (r^2-a^2) d\phi) \wedge dr \\
    \Sigma^3 &= \sqrt{\Delta} \sin(\theta) dt \wedge d\phi - \frac{\rho^2}{\sqrt{\Delta}} dr \wedge d\theta.
\end{align}
The self-dual connections obtained by solving $d^A \Sigma^i = d\Sigma^i + \epsilon^{ijk} A^j \wedge \Sigma^k = 0$ are the following 
\begin{align}
    A^1 &= \frac{M}{z_+^2} (dt - a\sin(\theta)^2 d\phi) + \frac{r\cos(\theta) - a}{z_+} d\phi \\
    A^2 &= -\frac{\sqrt{\Delta}\sin(\theta)}{z_+} d\phi \\
    A^3 &= -\frac{a\sin(\theta)}{\sqrt{\Delta} z_+} dr + \frac{\sqrt{\Delta}}{z_+} d\theta
\end{align}
where $z_+ = r-a \cos(\theta)$.
The self-dual curvatures can be computed (by hand!) and are given by
\begin{align}
    F^1 &= \frac{2M}{z_+^3}  \Sigma^1,\quad F^2 = -\frac{M}{z_+^3} \Sigma^2, \quad F^3 = -\frac{M}{z_+^3} \Sigma^3.
\end{align}
This verifies that Kerr is Ricci-flat, and also gives the components of the self-dual Weyl tensor $F^i = \Psi^{ij} \Sigma^j$ 
\begin{align}
    \Psi^{ij}_+ = \begin{pmatrix}
        2\alpha_+ & 0 & 0 \\
        0 & -\alpha_+ & 0 \\
        0 & 0 & -\alpha_+
    \end{pmatrix}, \quad \alpha_+ = \frac{M}{z_+^3}.
\end{align}
The matrix $\Psi^{ij}$ of the self-dual Weyl is tracefree as expected. Note that we have just used Plebanski formalism to verify that the Kerr metric is Ricci flat. This computation is doable by hand! This once again demonstrates the usefulness and power of the Plebanski formalism. 

\subsection{Anti-self-dual objects}

The anti-self-dual 2-forms can also be calculated using \cref{eq:AntiSelfDualTwoForm_From_Tetrad} to find
\begin{align}
    \bar{\Sigma}^1 &= (dt - a\sin(\theta)^2 d\phi) \wedge dr + \sin(\theta) d\theta \wedge ( a dt + (r^2-a^2) d\phi ) \\
    \bar{\Sigma}^2 &= \sqrt{\Delta}(dt-a\sin(\theta)^2 d\phi) \wedge d\theta +\frac{\sin(\theta)}{\sqrt{\Delta}}(a dt + (r^2-a^2) d\phi) \wedge dr \\
    \bar{\Sigma}^3 &= \sqrt{\Delta} \sin(\theta) dt \wedge d\phi + \frac{\rho^2}{\sqrt{\Delta}} dr \wedge d\theta.
\end{align}
The corresponding anti-self-dual connections are 
\begin{align}
    \bar{A}^1 &= -\frac{M}{z_-^2} (dt-a\sin(\theta)^2 d\phi) + \frac{r\cos(\theta) + a}{z_-} d\phi \\
    \bar{A}^2 &= -\frac{\sqrt{\Delta} \sin(\theta)}{z_-} d\phi\\
    \bar{A}^3 &= \frac{a \sin(\theta)}{\sqrt{\Delta} z} dr + \frac{\sqrt{\Delta}}{z_-} d\theta.
\end{align}
Here $z_-=r+a\cos(\theta)$.
The anti-self-dual part of the curvature is 
\begin{align}
    \bar{F}^1 = -\frac{2M}{z_-^3} \bar{\Sigma}^1, \quad \bar{F}^2 = \frac{M}{z_-^3} \bar{\Sigma}^2, \quad \bar{F}^3 = \frac{M}{z_-^3} \bar{\Sigma}^3.
\end{align}
Clearly the anti-self-dual Weyl tensor is 
\begin{align}
    \Psi_-^{ij} = \begin{pmatrix}
        2\alpha_- & 0 & 0 \\
        0 & -\alpha_- & 0 \\
        0 & 0 & -\alpha_-
    \end{pmatrix}, \quad \alpha_- = -\frac{M}{z_-^3}
\end{align}
Note that the anti-self-dual and self-dual objects are related by the transformation $t\rightarrow -t$, $a \rightarrow -a$, which is the Euclidean analogue of the complex conjugation in the Lorentzian signature.

\subsection{Conformal to K\"ahler}

Looking at both halves of the Weyl curvature we can see that Euclidean Kerr is of type $D\otimes D$. Then by Derdzi\'nski's theorem we know Euclidean Kerr is conformal to two different K\"ahler metrics.
The two conformal factors are 
\begin{align}
    \lambda_\pm = \frac{1}{z_\pm} \simeq \alpha_\pm^{\frac{1}{3}} \label{eq:KerrConformalFactors}
\end{align}
Conformally transforming the 2-forms $\Sigma^1$ and $\bar{\Sigma}^1$ we find a pair of closed 2-forms 
\begin{align}
    \omega_+ &= \frac{1}{z_+^2} \Sigma^1 = -dt \wedge d\left( \frac{1}{r-a\cos(\theta)} \right) + d\phi \wedge d\left( \frac{a-r\cos(\theta)}{r-a\cos(\theta)} \right) \\[10pt]
    \omega_- &= \frac{1}{z_-^2} \bar{\Sigma}^1 = dt \wedge d\left( \frac{1}{r+a\cos(\theta)} \right) - d\phi \wedge d \left( \frac{a + r\cos(\theta)}{r+a\cos(\theta)} \right).
\end{align}
The 2-forms $\omega_\pm$ are the K\"ahler forms of the two K\"ahler metrics.

\subsection{Complex Structures}

We can extract the complex structures from the metric and the corresponding K\"ahler forms.
It is useful to introduce a basis for $TM$ dual to the tetrad basis in \cref{eq:KerrTetradBasis}, these are given by
\begin{align}
    e_0 = \frac{r^2-a^2}{\rho \sqrt{\Delta}} \partial_t - \frac{a}{\rho \sqrt{\Delta}}  \partial_\phi, \quad e_1 = \frac{\sqrt{\Delta}}{\rho} \partial_r, \quad e_2 = \frac{1}{\rho}\partial_\theta, \quad e_3 = \frac{1}{\rho} \left( a \sin(\theta) \partial_t + \frac{1}{\sin(\theta)} \partial_\phi \right).
\end{align} 
These are dual to the tetrads in the sense that $e_I (e^J) = e^\mu_I e^J_\mu = \delta_I^J$.
The complex structures can be calculated at the level of the Einstein metric as they are conformally invariant.
Hence, we look for operators $J_\pm$ that satisfy $g(\cdot,J_\pm \cdot) = \lambda^{-2}_\pm \omega_\pm$.
Written in terms of tetrads and their duals this is simply
\begin{align}
    J_\pm &= e^1 \otimes e_0 - e^0 \otimes e_1 \mp e^3 \otimes e_2 \pm e^2 \otimes e_3.
\end{align}
It is not difficult to see that $J^2_\pm = - \id$. 
Substituting the expression in Boyer-Lindquist coordinates we find
\begin{align}
    J_\pm = &\frac{1}{\Delta} dr \otimes \left( (r^2-a^2) \partial_t  - a \partial_\phi \right) - \frac{\Delta}{\rho^2}(dt-a \sin(\theta)^2 d\phi) \otimes \partial_r \\
    & \mp \frac{\sin(\theta)}{\rho^2}(adt + (r^2-a^2)d\phi) \otimes \partial_\theta \pm d\theta \otimes \left( a \sin(\theta) \partial_t + \frac{1}{\sin(\theta)} \partial_\phi \right).
\end{align}
We have explicitly checked that these complex structures are integrable (by computing and checking that their Nijenhuis tensor vanishes) and that they commute. 

\subsection{Killing Vectors}

The Euclidean Kerr metric is conformal to two different K\"ahler metrics with opposite orientations. The general statements of the previous section imply that there should be two Killing vector fields, one coming from each of the two orientations. These are given by 
\begin{align}
    X_\pm^\mu = J_\pm^\mu{}_\nu (d\lambda_\pm^{-1})^\nu
\end{align}
Plugging in the conformal factors, \cref{eq:KerrConformalFactors}, we calculate the intermediate objects
\begin{align}
    d\lambda_\pm^{-1} = dr \pm a \cos(\theta) d\theta.
\end{align}
The vector fields dual to the 1-forms are given by 
\begin{align}
    (d\lambda_\pm^{-1})^\sharp = \frac{\Delta}{\rho^2} \partial_r \pm \frac{a \sin(\theta)}{\rho^2} \partial_\theta
\end{align}
We get the Killing vectors by applying the complex structures
\begin{align}
    X_\pm = J_\pm (d\lambda_\pm^{-1})^\sharp = -\partial_t.
\end{align}
Thus, we find that the two Killing vectors coming from the two orientations coincide in the case of Kerr. Thus, the Euclidean Kerr is covered by the case (iii) of Proposition 11 from \cite{AmbitoricGeomeAposto2013}. We note that the property that these two Killing vector fields coincide is related to the fact that the Kerr Laplace operator is separable. Indeed, this property holds only when one of the parameters of the Pleba\'nski-Demia\'nski solution is zero $\epsilon=0$, see below for the definition of $\epsilon$. This is also when the Laplace operator is separable. For more on these issues see e.g. \cite{Gauduchon_2017}.

Even though we see that for Kerr metric there is just one Killing vector field coming from the general construction of Proposition \ref{prop:ConfKahlerKillingVector}, there is still another Killing vector field $\partial_\phi$. The two Killing vector fields $\xi_1=\partial_t, \xi_2=\partial_\phi$ commute, and they also span an isotropic subspace $\omega_\pm(\xi_1,\xi_2)=0$ with respect to either of the K\"ahler forms.
It is also obvious that each of these two Killing vector fields is $J_\pm$-holomorphic, in the sense that ${\mathcal L}_{\xi_{1,2}} J_\pm=0$. This is manifest in the case of the Kerr metric, because all the geometric objects constructed for it are $t,\phi$ independent. We will take these properties of $\xi_1,\xi_2$ as the starting point of the derivation of the Kerr metric in the next section. 

For completeness, let us also compute the result of the action of the complex structures $J_\pm$ on $\xi_{1,2}$. We have 
\begin{gather}
    J_\pm(\partial_t) = \frac{\Delta}{\rho^2}\partial_r \pm \frac{a \sin(\theta)}{\rho^2} \partial_\theta \\
    J_\pm(\partial_\phi) = -\frac{a \sin(\theta)^2\Delta}{\rho^2} \partial_r \pm \frac{(r^2-a^2)\sin(\theta)}{\rho^2} \partial_\theta
\end{gather}
It is clear that $\xi_1,\xi_2, J(\xi_1),J(\xi_2)$ span all of $TM$, for each of the two complex structures. It can also be checked that all four of these vector fields are mutually commuting, for either $J$. The resulting frame will play an important role in the derivation of the next section.

\section{Toric double-sided type D metrics}\label{sec:DerivationOfTypeDSpacetimes}

We now proceed with a derivation of the Euclidean Kerr metric. In fact, our discussion will be more general, and include the more general Pleba\'nski-Demia\'nski family of metrics. 
The Pleba\'nski-Demia\'nski (PD) spacetimes \cite{RotatingChargPleban1976,ANewLookAtTGriffi2005} are Lorentzian algebraically special (type D) spacetimes that generalise many
important GR solutions. In general the PD spacetimes are solutions to the Einstein-Maxwell equations, however, here we will only consider the subclass of metrics that satisfy the vacuum Einstein equations. Our aim is to provide a "simple" derivation of this class of metrics by building on the conformal to K\"ahler ideas described in the previous section. 
Trying to setup Einstein's equation and solve them by brute force is doable but is far from the best strategy, see \cite{TheMathematicaChandr1998, PhysicallyMotiBaines2022, APossibleIntuNikoli2012} for an example of this. Instead, we follow the approach of \cite{Aposto2003,AmbitoricGeomeAposto2013} that is based on the availability of two different commuting complex structures in these spaces. 

\subsection{Properties of PD Spacetimes}

We begin by detailing properties of the spacetimes we aim to derive.
We are interested in Lorentzian signature spacetimes that are algebraically special (type D), although we will solve for their Euclidean analytic continuations.
For example, to obtain a Euclidean metric from the Kerr metric one sends $t \rightarrow it$ and $a \rightarrow ia$, where $t$ and $a$ are the time coordinate and the angular momentum parameter respectively. The self-dual (SD) and anti-self-dual (ASD) Weyl tensor in Lorentzian spacetimes are complex conjugates of each other, i.e. $W^+ = (W^-)^*$.
This is a consequence of the fact that SD 2-forms are complex-valued objects and the ASD 2-forms are complex conjugates of their SD counterparts.
Unlike the Euclidean signature case, there are no one sided type D metrics in Lorentzian signature. Thus, analytically continuing a type D Lorentzian metric to Euclidean signature we get a Riemannian metric that is type D with respect to both $W^\pm$. 

We are looking for Einstein metrics, for which we know both $W^\pm$ are divergence free. We can then resort to Derdzi\'nski's theorem \cref{thrm:DerdziConfKahler} which implies that the metric is conformal to two K\"ahler metrics, in general not coinciding.  We denote these two K\"ahler metrics by $g_\pm$, so that $g = \lambda^2_\pm g_\pm$ where $g$ is the Einstein metric of interest. Being conformal to two K\"ahler metrics implies the existence of two integrable almost complex structures $J_\pm$. Coming from the SD, ASD sectors respectively, these two complex structures have different orientations, and are not coinciding. Moreover, in 4D, two complex structures of different orientations commute. Having two commuting complex structures is an extremely strong property.  

By the discussion in the previous section, one-sided type D Einstein metrics are equipped with a Killing vector, see \cref{prop:ConfKahlerKillingVector}. In our situation of double-sided type D metrics we have two Killing vector fields of this type. However, they may coincide, which, as we have seen in the previous section, is indeed the situation in Kerr. However, the metric we want to reproduce has two commuting Killing vector fields, one related to its stationary property, the other related to it being axisymmetric. So, we will assume that there are two commuting Killing vectors $\xi_1,\xi_2$. In the Euclidean setting both of these Killing vectors have compact orbits. This means that we have the two-dimensional torus $\mathbb{T}^2$ acting on our space by isometries. Such spaces are called toric. Moreover, in our setup with two different K\"ahler metrics $g_\pm$, we have the action of $\mathbb{T}^2$ on both $g_\pm$. This is the reason why these spaces are called ambitoric in \cite{AmbitoricGeomeAposto2013}.

The final condition we impose is that the Killing vectors $\xi_1,\xi_2$ are $J$-holomoprhic with respect to both $J_\pm$, that is ${\mathcal L}_{\xi_{1,2}} J_\pm=0$, and that $\omega_\pm(\xi_1,\xi_2)=0$. We have seen that this is true in the case of Kerr, and so this is a well geometrically-motivated assumption. Having motivated the geometry of the problem, we are ready to convert the assumptions into a concrete metric ansatz. 

\subsection{A frame of commuting vector fields}

Let us consider just one of the two complex structures for the moment. The first statement that we would like to prove is that the collection of vector fields $\xi_1,\xi_2,J(\xi_1),J(\xi_2)$ spans $TM$ and mutually commutes. The first half of the statement follows from the fact that $\xi_1,\xi_2$ span an isotropic subspace $\omega(\xi_1,\xi_2)=0$. This means that $J(\xi_1),J(\xi_2)$ is the complementary subspace. 

To prove commutativity we need to use both J-holomorphicity, as well as integrability of $J$. Let us start by giving an equivalent way of stating J-holomorphicity. This is a standard statement in complex geometry.
\begin{lemma} A vector field $X$ is $J$-holomorphic ${\mathcal L}_X J=0$ if and only if $[X,JY]=J[X,Y], \forall Y\in TM$. 
\end{lemma}
\begin{proof} We have
\begin{align}
{\mathcal L}_X JY = ({\mathcal L}_X J)(Y) + J({\mathcal L}_X Y) \qquad \Rightarrow \qquad ({\mathcal L}_X J)(Y) = [X,JY]- J[X,Y].
\end{align}
This shows that $({\mathcal L}_X J)=0$ if and only if $[X,JY]- J[X,Y]=0$ for any $Y$.
\end{proof}

We can now use this to show that most of the vector fields in $\xi_1,\xi_2,J(\xi_1),J(\xi_2)$ commute. 
\begin{lemma} Let $[\xi_1,\xi_2]=0$ and $\xi_{1,2}$ be $J$-holomorphic. Then 
\begin{align} 
[\xi_1, J(\xi_1)] =0, \quad [\xi_1, J(\xi_2)] =0 \quad [\xi_2, J(\xi_1)] =0  \quad [\xi_2, J(\xi_2)] =0.
\end{align}
\end{lemma}
\begin{proof} This is obvious from the $J$-holomorphicity of $\xi_{1,2}$ and the previous lemma. 
\end{proof}

To prove that $J(\xi_1),J(\xi_2)$ commute, we need to show that the integrability of $J$ implies that they are also $J$-holomorphic. 
\begin{lemma} If $J$ is integrable, and $X$ is a $J$-holomorphic vector field, then $J(X)$ is also $J$-holomoprhic.
\end{lemma} 
\begin{proof}
$J$ is integrable if and only if the Nijenhuis tensor 
\begin{align}
N_J(X,Y) = [JX,JY]-J[X,JY]-J[JX,Y]-[X,Y]
\end{align}
vanishes. Let us assume that $X$ here is $J$-holomoprhic. Then the second and the fourth terms annihilate each other, and the vanishing of $N_J$ implies
\begin{align}
[JX,JY]=J[JX,Y],
\end{align}
which is precisely the statement that $JX$ is $J$-holomorphic.
\end{proof}

Now, with both $J(\xi_1),J(\xi_2)$ being $J$-holomorphic, it is clear that $[J(\xi_1),J(\xi_2)]=J^2 [\xi_1,\xi_2]=0$. We have established that $\xi_1,\xi_2,J(\xi_1),J(\xi_2)$ are mutually commuting. 

We end this subsection with another simple lemma.
\begin{lemma}
    Let $\xi_{1,2,3,4}$ span $TM$ and all commute. Let $\theta_{1,2,3,4}$ be the dual 1-forms. Then all these 1-forms are closed. 
\end{lemma}
\begin{proof}This follows from the standard formula
    \begin{align}
        d\theta(X,Y) = X(\theta(Y)) - Y(\theta(X)) - \theta([X,Y]), \quad \theta \in \Lambda^1, \quad X,Y \in TM.
    \end{align}
    Indeed, evaluating $d\theta$ for any of the 1-forms, on any pair of the vector fields $\xi_{1,2,3,4}$, we see that all the terms on the right-hand side are zero.  
   \end{proof}
   
 \subsection{Construction of the frame}

This subsection is central. Here we use the availability of the two commuting complex structures to find a convenient parametrisation of the 1-forms dual to $\xi_1,\xi_2,J(\xi_1),J(\xi_2)$. After such a parametrisation is obtained, one can characterise the two commuting complex structures very explicitly. Here we follow \cite{AmbitoricGeomeAposto2013} rather closely, and even keep some of the notations of these authors. However, unlike this reference, we try to be as elementary as possible, which means without using any algebraic geometry. 

First, we note that we can always choose the frame dual to $\xi_1,\xi_2,J(\xi_1),J(\xi_2)$ to be of the form $\theta_1,\theta_2,J(\theta_1),J(\theta_2)$. Since the original vector fields commute, these 1-forms are all closed. This means that there are local coordinates such that 
   \begin{align} 
   \theta_1 = d\tau, \qquad \theta_2= d\varphi.
   \end{align}
   The condition that the other two 1-forms are closed becomes
    \begin{align} 
   dJd\tau = 0, \qquad dJ d\varphi =0.
   \end{align}

We now come to the main statement of this subsection.
\begin{proposition} There exists a choice of coordinates $x,y$, and of two functions $F=F(x), G=G(y)$, such that the two commuting complex structures $J_\pm$ are realised as follows
\begin{align}
 J_\pm d\varphi &= \frac{1}{F}dx \pm \frac{1}{G} dy, \nonumber \\
    J_\pm d\tau &= \frac{x}{F}dx  \pm \frac{y}{G} dy, \nonumber\\
    J_\pm dx &= -\frac{F}{y-x}(d\tau - y d\varphi), \nonumber\\
    J_\pm dy &= \pm \frac{G}{y-x}(d\tau - x d\varphi). \label{eq:J_pm_action_on_txyf_basis}
\end{align}
\end{proposition}

\begin{proof}
Since $dJ_\pm d\tau = 0, dJ_\pm d\varphi = 0$, there exist coordinates $\xi,\eta$ such that 
\begin{gather}
    d\xi = \frac{1}{2}(J_++J_-)d\varphi, \quad d\eta = \frac{1}{2}(J_+-J_-)d\varphi. \nonumber \\
    J_\pm d\varphi= d\xi \pm d\eta
\end{gather}
The two coordinates introduced $(\xi,\eta)$ form a good coordinate system together with $\tau,\varphi$ because both complex structures map $(d\tau,d\varphi)$ to a complementary subspace, which we now parametrised by $(d\xi,d\eta)$. 

The action of the complex structures on $(d\xi,d\eta)$ gives a linear combination of $(d\tau,d\varphi)$, and can be parametrised as follows. 
\begin{gather}
    J_\pm d\xi = \alpha_\pm d\tau + \beta_\pm d\varphi, \quad J_\pm d\eta = \chi_\pm d\tau + \delta_\pm d\varphi.
\end{gather}
Here $\alpha_\pm,\beta_\pm,\chi_\pm,\delta_\pm$ are eight at this stage arbitrary functions. However, with $\partial_\tau,\partial_\varphi$ being Killing vectors, we can assume that these functions depend on $\xi,\eta$ only. Next we impose that $J_\pm^2 = -1$ on $d\varphi$
\begin{align}\label{J-pm-squared-eq}
    J_\pm^2 d\varphi = J_\pm d\xi \pm J_\pm d\eta = (\alpha_\pm \pm \chi_\pm) d\tau + (\beta_\pm \pm \delta_\pm) d\varphi = -d\varphi \nonumber \\
    \Rightarrow \quad \alpha_\pm \pm \chi_\pm = 0, \quad \beta_\pm \pm \delta_\pm =-1.
\end{align}
Next we enforce that $J_+$ and $J_-$ commute. We have
\begin{align}
    J_- J_+ d\varphi = J_- d\xi + J_- d\eta = (\alpha_- +\chi_-) d\tau +(\beta_- +\delta_-)  d\varphi, \nonumber \\
    J_+ J_- d\varphi = J_+ d\xi - J_+ d\eta = (\alpha_+ -\chi_+) d\tau +(\beta_+ -\delta_+)  d\varphi.
\end{align}
The commutativity $[J_+,J_-]=0$, together with the conditions in \cref{J-pm-squared-eq}, implies that the action of $J_\pm$ on $d\xi,d\eta$ can be parametrised by only two functions
\begin{gather}
    J_\pm d\xi = -\chi d\tau - (1+\delta) d\varphi, \qquad J_\pm d\eta = \pm (\chi d\tau + \delta d\varphi).
\end{gather}
Lastly, we look at the action on $d\tau$, in general it will have the form 
\begin{align}
    J_\pm d\tau = a_\pm d\xi + b_\pm d\eta \nonumber
\end{align}
with 4 arbitrary functions $a_\pm, b_\pm$. Again we check
\begin{gather}
    J_\pm^2 d\tau = a_\pm J_\pm d\xi + b_\pm J_\pm d\eta = \left( -a_\pm  \pm b_\pm \right) \chi d\tau +  \left( - a_\pm (1+ \delta) \pm b_\pm  \delta\right) d\varphi= -d\tau \nonumber \\
    \Rightarrow \quad a_\pm = -\frac{\delta}{\chi}, \quad b_\pm = \mp \frac{1+\delta}{\chi}
\end{gather}
Hence the action on $d\tau$ is 
\begin{gather}
    J_\pm d\tau = - \frac{\delta}{\chi} d\xi \mp \frac{1+\delta}{\chi} d\eta
\end{gather}
We have already imposed the closure $d J_\pm d\varphi=0$. It remains to impose $d J_\pm d\tau=0$. We get
\begin{gather}
  0=  d(J_\pm d\tau) = \left[ \left( \frac{\delta}{\chi} \right)_\eta \mp \left( \frac{1+\delta}{\chi} \right)_\xi \right] d\xi \wedge d\eta \nonumber \\
    \Rightarrow \quad \left( \frac{\delta}{\chi} \right)_\eta = 0, \quad \left( \frac{1+\delta}{\chi} \right)_\xi = 0
\end{gather}
These equations are easily solved by introducing two functions $a = a(\xi)$, $b = b(\eta)$ such that $\delta = -\frac{a}{a-b}, \chi = \frac{1}{a-b}$.
Summarising the results so far we have
\begin{align}
    J_\pm d\varphi =& d\xi \pm d\eta \nonumber \\
    J_\pm d\tau =& a d\xi \pm b d\eta \nonumber \\
    J_\pm d\xi =& \frac{1}{a-b} \left( - d\tau + b d\varphi \right) \nonumber \\
    J_\pm d\eta =& \pm\frac{1}{a-b}\left(  d\tau -a d\varphi \right)
\end{align}
From here we make a coordinate transformation $x = a(\xi), y = b(\eta)$. Writing $dx = a' d\xi \equiv F d\xi, dy = b' d\eta= G d\eta$, 
we rewrite the action of the complex structures as
\begin{align}
    J_\pm d\varphi =& \frac{dx}{F} \pm \frac{dy}{G} \nonumber \\
    J_\pm d\tau =& \frac{x dx}{F} \pm \frac{y dy}{G} \nonumber \\
    J_\pm dx =& - \frac{F}{x-y}\left(  d\tau - y d\varphi \right) \nonumber \\
    J_\pm dy =& \pm \frac{G}{x-y}\left(  d\tau - x d\varphi\right),
\end{align}
where $F=F(x), G=G(y)$. 

\end{proof}

\subsection{Product structure}

We now consider the operator $j:= -J_+ J_-$. This squares to the identity $j^2 = \id$. It is also clear that, because both $J_\pm$ are orthogonal (i.e. metric-compatible) complex structures, the operator $j$ is also metric-compatible. Such an operator $j:TM\to TM: j^2=\id$ is known as an orthogonal product structure. The reason for this terminology is that 
it decomposes the (tangent and) cotangent space into the eigenspaces of eigenvalue $\pm 1$. It is easy to see that the subspaces of opposite eigenvalue are metric orthogonal. Indeed, let $\eta_1, \eta_2\in \Lambda^1$ be arbitrary eigenforms of $j$ belonging to different eigenspaces 
    \begin{align}
        j \eta_+ = +\eta_+, \quad j \eta_- = - \eta_-.
    \end{align}
   We then have
    \begin{align}
        (\eta_-, \eta_+) = (j \eta_-, j \eta_+) = -(\eta_-, \eta_+) = 0.
    \end{align}
    Here the bracket denotes the metric pairing. 
   
   Now, using one of the complex structures, say $J_+$, we can introduce two more 1-forms
    \begin{align}
        J_+ \eta_-, J_+ \eta_+ \in \Lambda^1, \quad s.t. \quad j(J_+ \eta_-) = - J_+ \eta_-, \quad j(J_+ \eta_+) = +J_+ \eta_+.
    \end{align}
    It is easy to see that $(J_+ \eta_-,\eta_-)=0, (J_+ \eta_+,\eta_+)=0$ and so the basis of 1-forms 
        \begin{align}
        \eta_+,\quad J_+ \eta_+,\quad \eta_-,\quad J_+ \eta_-
    \end{align}
   is metric-orthogonal. 

\subsection{Metric ansatz}

Using the complex structures defined in \cref{eq:J_pm_action_on_txyf_basis}, it is now an easy exercise to check that
\begin{align}
j(dx)= dx, \qquad j(dy)=- dy.
\end{align}
This means that the 1-forms 
\begin{align}
dx, dy, (d\tau - y d\varphi), (d\tau- x d\varphi),
\end{align}
where the last two 1-forms are obtained by applying e.g. $J_+$ to $dx,dy$, form a metric orthogonal basis. This means that the metric we are considering is of the form
\begin{align}
    g = A(d\tau-y d\varphi)^2 + B dx^2 + C dy^2 + D(d\tau- x d\varphi)^2. \label{eq:General_Ambitoric_Metric}
\end{align}
where $A,B,C,D$ are functions of $x,y$ only.

However, metrics of the form (\ref{eq:General_Ambitoric_Metric}) are not always compatible with the 2 complex structures.
To require this we impose that $g(\cdot ,J_\pm \cdot) \in \Lambda^2$.
Computing the tensor $g(\cdot, J_\pm \cdot)$ we find
\begin{align}
    g(\cdot, J_\pm \cdot) = \frac{A(x-y)}{F}dx \otimes (dt - y d\varphi) - BF\frac{(dt-y d\varphi)}{x-y} \otimes dx \\ \nonumber
    \mp \frac{D(x-y)}{G}dy \otimes (dt - x d\varphi) \pm CG \frac{(dt - x d\varphi)}{x-y} \otimes dy.
\end{align}
Imposing that this is a 2-form results in the following restrictions to the components of $g$,
\begin{align}
    \frac{A(x-y)}{F} = \frac{BF}{x-y}, \quad \frac{D(x-y)}{G} = \frac{CG}{x-y}.
\end{align}
Which are most conveniently solved by introducing two new functions $U = U(x,y)$ and $V = V(x,y)$ so that the solutions become
\begin{align}
    A = \frac{FU}{(x-y)^2}, \quad B = \frac{U}{F}, \quad C = \frac{V}{G}, \quad D = \frac{GV}{(x-y)^2}.
\end{align}
Thus, finally, the metric and the two compatible 2-forms ($\Sigma^1_\pm = g(\cdot, J_\pm \cdot$))  are given by
\begin{align}
    g = FU \left( \frac{d\tau - y d\varphi}{x-y} \right)^2 + \frac{U}{F} dx^2 + \frac{V}{G} dy^2 + GV \left( \frac{d\tau-xd\varphi}{x-y} \right)^2 \\
    \Sigma^1_\pm = \frac{U}{x-y} dx \wedge (dt - yd\varphi) \mp \frac{V}{x-y} dy \wedge (d\tau - x d\varphi)\label{eq:Conformal2Form_and_ConformalMetric}.
\end{align}
It should be appreciated how close this metric ansatz is to the Euclidean Kerr metric. We emphasise that the only information that went into the construction of this ansatz is that there are two commuting Killing vector fields, as well as two commuting complex structures. No field equations have yet been imposed. 

\subsection{Conformal to two K\"ahler metrics}

We now want to impose the requirement that $g$ is conformal to two different K\"ahler metrics. That is, we demand that $\Sigma^1_\pm$ are conformal to closed 2-forms 
\begin{align}
    d( \lambda^2_\pm \Sigma^1_\pm) = 0.
\end{align}
with $\lambda_\pm = \lambda_\pm(x,y)$ being the conformal factors for each K\"ahler metric.
The Lee forms of $\Sigma^1_\pm$ are 1-forms, $\theta_\pm$, defined by 
\begin{align}
    d\Sigma^1_\pm = \theta_\pm \wedge \Sigma^1_\pm.
\end{align}
When $\Sigma^1_\pm$ are conformal to closed 2-forms the Lee forms themselves are closed.
In terms of the conformal factors we find that $\theta_\pm = - d\ln(\lambda^2_\pm)$, and hence closed.
Using the 2-forms in \cref{eq:Conformal2Form_and_ConformalMetric} and solving for $\theta_\pm$ we find them to be
\begin{align}
    \theta_\pm = \left( \frac{V_x}{V} \pm \frac{U}{V(x-y)} \right) dx + \left( \frac{U_y}{U} \mp \frac{V}{U(x-y)} \right)dy,
\end{align}
where we write $V_x = \partial_x V$ and similarly for $y$. We want to impose the condition that both of these are closed 1-forms. 
It is more efficient to take linear combinations. Firstly, we demand 
\begin{align}
    d(\theta_+ + \theta_-) = 0 \quad \Rightarrow \quad \ln(U)_{xy} - \ln(V)_{xy} = 0
\end{align}
which can be solved by reparameterising by new functions $A = A(x),\ B = B(y)$ and $H = H(x,y)$, such that
\begin{align}
    U = \frac{x-y}{H^2A}, \quad V = \frac{x-y}{H^2B}.
\end{align}
The introduction of $x-y$ here is to help simplify later formulae.
The other linear combination is 
\begin{align}
    d(\theta_+ - \theta_-) = 0 \quad \Rightarrow \quad -(B^2)_y (x-y) - 2B^2 = (A^2)_x (x-y) - 2A^2. \label{eq:AB_Quadratic_Equation}
\end{align}
Taking derivatives twice with respect to $x$ we get $(A^2)_{xxx} = 0$. Similarly taking the derivative with respect to $y$ twice we get  $(B^2)_{yyy}=0$. This means that 
these equations can be solved by $A^2$ and $B^2$ being quadratic polynomials in their respective variables, with arbitrary coefficients. Further, 
substituting the quadratic ansatz for $A^2,B^2$ back into \cref{eq:AB_Quadratic_Equation} we find that their coefficients coincide, that is 
\begin{align}
    A^2 = R(x), \quad B^2 = R(y),\quad R(z) = r_0 + r_1 z + r_2 z^2. \label{eq:AmbiToric_AB_R_Solution}
\end{align}
We thus find that, given the two functions $F(x),G(y)$ and the 3 constants $r_0,r_1,r_2$ defining $A(x),B(y)$, the metric 
\begin{align}
    g = \frac{1}{H^2} \left[ \frac{F}{A(x-y)} (d\tau-yd\varphi)^2 + \frac{x-y}{FA} dx^2 + \frac{x-y}{BG} dy^2 + \frac{G}{B(x-y)} (d\tau-x d\varphi)^2 \right] \label{eq:GeneralAmbitoric_Metric}
\end{align}
is an ambitoric metric that is conformal to two K\"ahler metrics with opposite orientations. This ansatz for the metric is half-way to determining the Kerr metric.

\subsection{Solving for the conformal factor}

We are interested in a metric that is Einstein, so we now impose Einstein's equations. 
The easiest way to do this is using the chiral formalism explained in \cref{sec:Chiral_Formalism}. The derivation we present here is different from that in \cite{AmbitoricGeomeAposto2013}, and, we hope, more straightforward. From \cref{eq:GeneralAmbitoric_Metric} we choose a convenient basis for the frame 
\begin{align}
    e^0 = \frac{1}{H}\sqrt{\frac{F}{A(x-y)}}(d\tau-yd\varphi), \, e^1 = \frac{1}{H}\sqrt{\frac{x-y}{FA}}dx, \, e^2 = \frac{1}{H}\sqrt{\frac{x-y}{BG}}dy, \, e^3 = \frac{1}{H}\sqrt{\frac{G}{B(x-y)}} (d\tau - x d\varphi).
\end{align}
Using \cref{eq:SelfDualTwoForms_From_Tetrad} we can build the SD 2-forms
\begin{align}
    \Sigma^1 &= \frac{1}{H^2} \left( \frac{1}{A} (d\tau-yd\varphi) \wedge dx - \frac{1}{B} dy \wedge (d\tau-xd\varphi) \right) \\
    \Sigma^2 &= \frac{1}{H^2} \left( \sqrt{\frac{F}{ABG}} (d\tau-y d\varphi) \wedge dy - \sqrt{\frac{G}{ABF}} (d\tau-x d\varphi) \wedge dx \right) \\
    \Sigma^3 &= \frac{1}{H^2} \left( \sqrt{\frac{FH}{AB}} d\varphi\wedge d\tau - \frac{x-y}{\sqrt{ABF}} dx \wedge dy \right).
\end{align}
The self-dual connection, obtained by solving $d^A \Sigma^i = d\Sigma^i + \epsilon^{ijk} A^j \wedge \Sigma^k = 0$, in this basis becomes
\begin{align}
    A^1 &= \left( \frac{BF}{A} - \frac{F^2}{AH^2} \left( \frac{(x-y)AH^2}{F} \right)_x \right) \frac{d\tau-y d\varphi}{2(x-y)^2} - \left( \frac{AG}{B} + \frac{G^2}{BH^2} \left( \frac{(x-y)BH^2}{G} \right)_y \right) \frac{d\tau - x d\varphi}{2(x-y)^2} \nonumber\\
    A^2 &= \sqrt{\frac{FG}{AB}} \left[  \frac{A H_x - B H_y}{H(x-y)} d\tau + \frac{H(A-B) + 2By H_y - 2Ax H_x}{H(x-y)} d\varphi \right] \label{eq:AmbiToric_Connection_xy}\\
    A^3 &= \frac{1}{\sqrt{AB}} \left[ \sqrt{\frac{F}{G}} \frac{H(B-A) + 2B(x-y)H_y}{2H(x-y)} dx + \sqrt{\frac{G}{F}} \frac{H(A-B)-2A(x-y) H_x}{2H(x-y)} dy \right].\nonumber
\end{align}

The curvatures are given by rather long expressions. This is the only place where we needed to resort to algebraic manipulation software. We state the curvature 2-forms by decomposing them into their self-dual and anti-self dual components. We write 
\begin{align}
    F^i = M^{ij} \Sigma^j + R^{ij} \bar{\Sigma}^j.
\end{align}
The nonzero components of $M^{ij}$ are 
\begin{align}
    M^{11} =&\frac{AH^2}{24}\left( \frac{F\log(A)_x}{x-y} \right)_x - \frac{AH^3}{12(x-y)} \left( F\left( \frac{1}{H} \right)_x \right)_x - \frac{AH^2}{24}\left( \frac{F}{x-y} \right)_{xx} - \frac{AFH^2}{48(x-y)^3} \\
            &+\frac{BH^2}{24}\left( \frac{G\log(B)_x}{x-y} \right)_y - \frac{BH^3}{12(x-y)} \left( G\left( \frac{1}{H} \right)_y \right)_y - \frac{BH^2}{24}\left( \frac{G}{x-y} \right)_{yy} - \frac{BGH^2}{48(x-y)^3} \nonumber\\
            &+ \frac{ABH^2}{12}\left[ \left( \frac{G}{B(x-y)^2} \right)_y - \left( \frac{F}{A(x-y)^2} \right)_x \right] - \frac{H^2}{16(x-y)^3}\left( \frac{A^3 G + B^3 F}{AB} \right) \nonumber \\[10pt]
    M^{22} = M^{33} =& \frac{ABH^2}{24}\left( \frac{F}{A(x-y)^2} \right)_x - \frac{ABH^2}{24}\left( \frac{G}{B(x-y)^2} \right)_y + \frac{H^2}{48(x-y)^3} \left( \frac{A^3G+B^3F}{AB} \right) \\
                     & - \frac{AH^3}{12(x-y)}\left( F\left( \frac{1}{H} \right)_x \right)_x - \frac{AFH^2}{48(x-y)^3} - \frac{AH^2}{24} \left( \frac{F}{(x-y)^2} \right)_x \nonumber\\
                     & - \frac{BH^3}{12(x-y)}\left( G\left( \frac{1}{H} \right)_y \right)_y - \frac{BGH^2}{48(x-y)^3} + \frac{BH^2}{24} \left( \frac{G}{(x-y)^2} \right)_y \nonumber \\[10pt]
    M^{12} = M^{21} =& -\sqrt{\frac{FG}{AB}}\frac{H^2}{24(x-y)^3} \left( \frac{x-y}{2}(A^2)_x - A^2 +\frac{x-y}{2}(B^2)_y + B^2 \right)
\end{align}
We remind that one of the Einstein equations states that the trace of the matrix $M^{ij}$ coincides with the cosmological constant $M^{11}+M^{22}+M^{33}=\Lambda$. The components of $R^{ij}$ are
\begin{align}
    R^{11} =& \frac{AFH}{12}\left( \frac{H}{x-y} \right)_{xx} - \frac{AH^2}{24}\left( \frac{F}{x-y} \right)_{xx} + \frac{AHF_x H_x}{12(x-y)} +\frac{AH^2}{24}\left( \frac{F\log(A)_x}{x-y} \right)_x - \frac{7AFH^2}{48(x-y)^3}\\
            & -\frac{BGH}{12}\left( \frac{H}{x-y} \right)_{yy} - \frac{BH^2}{24}\left( \frac{G}{x-y} \right)_{yy} + \frac{BHG_y H_y}{12(x-y)} +\frac{BH^2}{24}\left( \frac{G\log(B)_y}{x-y} \right)_y - \frac{7BGH^2}{48(x-y)^3}\nonumber \\
            & +\frac{H^2}{16(x-y)^3}\left( \frac{B^3 F - A^3 G}{AB} \right) \nonumber \\[10pt]
    R^{22} =& -\frac{FH(AH_x)_x}{24(x-y)} + \frac{FH^2A_x}{24(x-y)^2} - \frac{AFH^2}{48(x-y)^3} +\frac{GH(BH_y)_y}{24(x-y)} + \frac{GH^2B_y}{24(x-y)^2} - \frac{BGH^2}{48(x-y)^3} \\
            & + \frac{H^2}{48(x-y)^3} \left( \frac{B^3 F - A^3 G}{AB} \right) \nonumber \\[10pt]
    R^{33} =& -\frac{FH(AH_x)_x}{24(x-y)} + \frac{FH^2A_x}{24(x-y)^2} - \frac{AFH^2}{48(x-y)^3} -\frac{GH(BH_y)_y}{24(x-y)} - \frac{GH^2B_y}{24(x-y)^2} + \frac{BGH^2}{48(x-y)^3} \\
            & - \frac{H^2}{48(x-y)^3} \left( \frac{B^3 F - A^3 G}{AB} \right) \nonumber \\[10pt]
    R^{12} =& \sqrt{\frac{FG}{AB}}\frac{H}{12(x-y)^2} \left[ 2AB\left( 2(x-y)H_{xy} + H_x - H_y \right) + \left( 2B^2 H_y - \frac{H}{2}(B^2)_y - 2A^2 H_x + \frac{H}{2}(A^2)_x \right) \right] \\
    R^{21} =& \sqrt{\frac{FG}{AB}}\frac{H}{12(x-y)^2} \left[ 2AB\left( 2(x-y)H_{xy} + H_x - H_y \right) - \left( 2B^2 H_y - \frac{H}{2}(B^2)_y - 2A^2 H_x + \frac{H}{2}(A^2)_x \right) \right]
\end{align}

We find it useful to first analyse the equations imposed by $R^{12} = 0, R^{21}=0$. Taking the sum and difference of these equations we get
\begin{align}
    2(x-y)H_{xy} + H_x - H_y = 0 \label{eq:H_equation_xy} \\
    2A^2 H_x - \frac{H}{2}(A^2)_x -2B^2 H_y + \frac{H}{2}(B^2)_y = 0. \label{eq:HAB_equation_xy}
\end{align}
To solve the first equation we introduce a coordinate substitution $t,r = (x+y,x-y)$ such that is becomes 
\begin{align}
    H_{tt} = r\left( \frac{1}{r}H_r \right)_r.
\end{align}
Looking for separable solutions of the form $H = T(t) + R(r)$ we obtain the following class of solutions
\begin{align}
    H(t,r) = c_0 r^2 \ln(r) + c_3 r^2 + c_0 t^2  + c_1 t + c_2.
\end{align}
Another class of solutions is obtained by taking $H = H(z)$ and $z = \sqrt{t^2-r^2}$ the equation reduces to 
\begin{align}
    H_{zz} = 0
\end{align}
which has a solution of the form 
\begin{align}
    H(t,r) = c_4 \sqrt{t^2-r^2} + const.
\end{align}
As the original PDE is linear in $H$, we can add these solutions.
We note that other separable solutions exist in terms of Bessel functions, but these will not be relevant.
Substituting the solution into \cref{eq:HAB_equation_xy} we can restrict the allowed values of the constants $c_i$ and $r_i$, see \cref{eq:AmbiToric_AB_R_Solution} for the definition of the latter. In doing so we find 3 different solutions for $H(x,y)$ and $R(z)$
\begin{alignat}{2}
    1) \quad H(x,y) &= 1 + \varepsilon \sqrt{xy}, \quad &&R = z \nonumber\\
    2) \quad H(x,y) &= \sqrt{xy}, \quad &&R = z^2 + 2\varepsilon z \label{eq:HR_solutions_cases} \\
    3) \quad H(x,y) &= 1 + \varepsilon (x+y), \quad &&R = 1 \nonumber
\end{alignat}
where we have renamed some constants for later convenience.
In the main text we will focus on case $1)$ as this gives the family that contains the PD metrics.
The other cases are described in \cite{AmbitoricGeomeAposto2013} and briefly in \cref{Apdx:Alternative_Cases_Solution}.
In case $1)$, we have that the solutions for $H,A,B$ are
\begin{align}
    H = 1+\varepsilon \sqrt{xy}, \quad A(x) = \sqrt{x}, \quad B(y) = \sqrt{y}.
\end{align}
We now perform another change of variables to remove the square roots. We take $x = r^2,\ y = q^2$, taking the positive branch of the square root.
The metric in these variables becomes 
\begin{align}
    g = \frac{1}{(1+\varepsilon rq)^2} \left( \frac{C}{r^2-q^2}(d\tau-q^2 d\varphi)^2 + \frac{r^2-q^2}{C} dr^2 + \frac{r^2-q^2}{D}dq^2 + \frac{D}{r^2-q^2} (d\tau-r^2 d\varphi)^2 \right) \label{eq:PlebanskiDemianskiMetric}
\end{align}
where we have introduced $C(r) = F/r$, $D(q) = G/q$.

\subsection{Pleba\'nski-Demia\'nski family of solutions}

Having chosen a solution for the conformal factor function $H^2$, we now calculate the Ricci scalar condition $M^{11}+M^{22}+M^{33} = \Lambda$. The components of the matrix $M^{ij}$ for the metric \cref{eq:PlebanskiDemianskiMetric} are as follows
\begin{align}
    M^{11} =& \frac{1}{6}\frac{(1+\varepsilon r q)^2(r+q)^3}{q^2-r^2}\left[ \left( \frac{C}{(r+q)^3} \right)_{rr} \hspace{-10pt}+ \left( \frac{D}{(r+q)^3} \right)_{qq} \right] \\ \nonumber
    +& \frac{(1+\varepsilon r q)^5}{12(q^2-r^2)} \left[ \left( \frac{C}{(1+\varepsilon rq)^3} \right)_{rr} \hspace{-10pt} + \left( \frac{D}{(1+\varepsilon rq)^3} \right)_{qq} \right], \\
    M^{22} = M^{33} =& -\frac{1}{12}\frac{(1+\varepsilon r q)^2(r+q)^3}{q^2-r^2}\left[ \left( \frac{C}{(r+q)^3} \right)_{rr} \hspace{-10pt}+ \left( \frac{D}{(r+q)^3} \right)_{qq} \right] \\ \nonumber
    + & \frac{(1+\varepsilon r q)^5}{12(q^2-r^2)} \left[ \left( \frac{C}{(1+\varepsilon rq)^3} \right)_{rr} \hspace{-10pt} + \left( \frac{D}{(1+\varepsilon rq)^3} \right)_{qq} \right].
\end{align}
Taking the trace results in the following equation
\begin{align}
    12 \varepsilon^2 (q^2 C + r^2 D) - 6\varepsilon(1+\varepsilon rq)(qC_r + r D_q) + (1+\varepsilon rq)^2 (C_{rr} + D_{qq}) = 4 (q^2-r^2) \Lambda.
\end{align}
By taking derivatives of both sides 3 times with respect to $r$ and $q$ gives 
\begin{align}
    C_{rrrrr} = 0, \quad D_{qqqqq} = 0.
\end{align}
This shows that $C$ and $D$ are now 4th order polynomials of their respective variables.
The only remaining nonzero component of the traceless Ricci tensor is then
\begin{align}
    R^{11} =& (r^2-q^2)(1+\varepsilon rq) D_{qq} + (4q+6\varepsilon rq^2 - 2\varepsilon r^3) D_q -4(1+3\varepsilon rq) D \nonumber\\
            & -(r^2-q^2)(1+\varepsilon rq) C_{rr} + (4r+6\varepsilon r^2q - 2\varepsilon q^3) C_r -4(1+3\varepsilon rq) C.
\end{align}
This allows to further restrict $C,D$. We substitute a generic quartic polynomial ansatz for C,D into $M^{11}+M^{22}+M^{33} = \Lambda$ and $R^{11} = 0$.
This gives the following solution
\begin{align}\label{eq:PlebanskiDemianskiFinalSolution}
    C = -b + 2mr +  r^2 + 2n\varepsilon r^3 - (\varepsilon^2 b + \Lambda/3) r^4 \\
    D = b + 2n q - q^2 + 2m\varepsilon q^3 + (\varepsilon^2 b + \Lambda/3) q^4.
\end{align}
We have use the freedom of rescaling $r,q \rightarrow \lambda r, \lambda q$ to set the coefficients in front of $r^2, -q^2$ to unity.

We have thus derived the Euclidean version of the Pleba\'nski-Demia\'nski metrics with their 5 parameters, $b,m,n,\varepsilon,\Lambda$ \cite{RotatingChargPleban1976,ANewLookAtTGriffi2005}. This is the same result that was obtained in \cite{AmbitoricGeomeAposto2013}, however the approach here was, we believe, more elementary, without any need for algebraic geometry considerations. 

To check that the metric we obtained is indeed one sided type $D \otimes D$ we can calculate both parts of the Weyl curvature. We get
\begin{align}
    \Psi = \begin{pmatrix}
        2\alpha & 0 & 0 \\
        0 & -\alpha & 0  \\
        0 & 0 & -\alpha
    \end{pmatrix}, \quad \alpha = \frac{n-m}{6} \left( \frac{1+\varepsilon rq}{r+q} \right)^3, \\ \nonumber
    \bar{\Psi} = \begin{pmatrix}
        2\bar{\alpha} & 0 & 0 \\
        0 & -\bar{\alpha} & 0 \\
        0 & 0 & -\bar{\alpha}
    \end{pmatrix}, \quad \bar{\alpha} = \frac{n+m}{6} \left( \frac{1+\varepsilon rq}{r-q} \right)^3.
\end{align}
Clearly both halves of the Weyl tensor are of type D.
When $m = \pm n$ one of the sides of the Weyl curvature vanishes, in this case the spacetimes are called self-dual Pleba\'nski-Demia\'nski \cite{KillingYanoTeNozawa2015}.

\subsection{Kerr Metric}\label{sec:PDMetricToKerrMetric}

To obtain the Kerr metric from the more general Pleba\'nski-Demia\'nski family of metrics we set $\Lambda = 0$. We want the metric to be asymptotically flat, requires that $g_{rr} \rightarrow 1$ as $r \rightarrow \infty$. For this to be the case the polynomial $C(r)$ must be at most quadratic, which necessitates $\varepsilon = 0$.
Now both $C,D$ are simple quadratic functions
\begin{align}
    C(r) = r^2 - 2Mr -a^2, \qquad 
    D(q) = a^2 - q^2,
\end{align}
where we renamed the constant terms suggestively, and also set the NUT charge (vanishing in Kerr) $n=0$. Now $C(r)=\Delta$ of the Kerr metric, and setting $q=a\cos\theta$ makes $D=a^2\sin^2\theta$. These choices can be motivated using the requirement that the metric is an appropriate analytic continuation of a Lorentzian signature metric, and that the Lorentzian metric is axi-symmetric, but we will not attempt this here. The main claim we are making is that the Euclidean Kerr metric is a member of a large family of metrics whose determination reduces to an exercise in linear algebra, determining the metric ansatz, following by a straightforward computation of curvatures using the Plebanski formalism. 

With these choices the metric \cref{eq:PlebanskiDemianskiMetric} becomes
\begin{align}
    g = \frac{\Delta}{\rho^2}(d\tau- a^2\cos^2\theta d\varphi)^2 + \frac{\rho^2}{\Delta} dr^2 + \rho^2 d\theta^2 + \frac{\sin^2\theta}{\rho^2} (a d\tau-a r^2 d\varphi)^2.
 \end{align}
 Now further changing the coordinates
 \begin{align}
 d\varphi = -\frac{1}{a} d\phi, \qquad d\tau = dt - a d\phi
 \end{align}
 gives the usual form of the Kerr metric in Boyer-Lindquist coordinates, with the frame given by \cref{eq:KerrTetradBasis}. The $\Lambda\not=0$ Kerr metric is another member of the Pleba\'nski-Demia\'nski family, but we have not attempted to exhibit it explicitly.

\section{Conclusion}

We have presented an alternative derivation of the (Euclidean) Kerr metric, one that proceeds by deriving the more general Pleba\'nski-Demia\'nski family. The derivation we described is "elementary" in the sense that the most non-trivial step of the construction, which is establishing the ansatz (\ref{eq:Conformal2Form_and_ConformalMetric}), proceeds by a calculation in linear algebra. Thus, one gets rather far without imposing any differential equations whatsoever. Some differential equations are imposed at the next step, which requires that the metric is conformal to two different K\"ahler metrics. This leads to the metric ansatz (\ref{eq:GeneralAmbitoric_Metric}). The rest of the analysis consists in imposing Einstein equations. This is seen to reduce the arbitrariness to two quartic polynomials, each of one variable. One can then recognise the Euclidean Kerr metric in the obtained family of solutions without difficulty. It corresponds to the case of quadratic polynomials. 

Our main reason to be interested in the geometry we described is the fact that the (Euclidean) Kerr metric exhibits the beautiful and rich geometry if viewed as a complex 4D manifold.  This is of course also one of the motivations of \cite{AmbitoricGeomeAposto2013}. It can also be seen from the derivation we presented that the Plebanski chiral formalism is best-suited for viewing 4D Euclidean geometry via the prism of complex geometry. Most of what we described is not new for mathematicians, apart from our new characterisation of Killing vector fields in section 3.2 and the new proof of Derdzi\'nski theorem in section 4. Some of the reasoning in section 6, in particular in the part where we impose Einstein equations, is also different from \cite{AmbitoricGeomeAposto2013}. But our main motivation in describing these results is our hope that the adopted here point of view of the Plebanski formalism will make these beautiful ideas more digestible for the gravitational physicists. 

Our final remark is that it is very likely that the described here rich ambi-K\"ahler and ambi-toric geometry of the Kerr metric is useful for the problem of gravitational perturbations of the Kerr BH. It is well-known that the problem of scalar, vector and tensor perturbations in the Kerr background is separable, see e.g. \cite{TheMathematicaChandr1998}. This is usually achieved by considering the Teukolsky equations. On the other hand, it is also well-known that the separability of the Kerr metric wave equation is due to the existence of a non-trivial symmetric rank 2 Killing tensor. This Killing tensor can be constructed as the square of the anti-symmetric rank 2 Killing-Yano tensor. All this is intimately related to the construction described in this paper, as is explained in \cite{Gauduchon_2017}. Indeed, in this reference the authors show that the existence of a Killing 2-form in 4D is directly related to the metric being ambi-K\"ahler. All in all, the described here geometry is known to be directly related to the separability of the Kerr wave equation. It thus appears likely that there exists a formalism for gravitational perturbation theory in Kerr that is based on the Plebanski formalism, and the geometry described in this paper. This could be distinct from that given by the Teukolsky equation, and potentially more useful. The basic point is that the Teukolsky equation reduces the problem of describing gravitational perturbations of Kerr to a single scalar. This scalar is an appropriate potential for the metric. It is possible that the geometry we described can lead to a different choice of a potential for the metric. In particular, the pure connection formalism for 4D GR, see \cite{Krasnov:2011pp} may be useful here, in that a potential for the connection rather than the metric may give a more powerful description. We hope to return to such questions in a future publication.

\appendix
\section{Alternative Cases Solution}\label{Apdx:Alternative_Cases_Solution}
Here we briefly explore the Einstein metrics that are solutions to the other cases in \cref{eq:HR_solutions_cases}.

\subsection{Case 2}
In case 2 we can make a slight change of coordinates $r = x+\varepsilon, q = y+\varepsilon$ such that 
\begin{align}
    H = \sqrt{(r-\varepsilon)(q-\varepsilon)}, \quad A = \sqrt{r^2 - \varepsilon^2}, \quad B = \sqrt{q^2 - \varepsilon^2}.
\end{align}
We redefine the functions in the metric $C = F A$, $D = G B$ such that the metric becomes 
\begin{align}
    g = \frac{1}{(r-\varepsilon)(q-\varepsilon)}\left( C\frac{(d\tau-(q-\varepsilon)d\phi)^2}{(r^2-\varepsilon^2)(r-q)} + \frac{r-q}{C} dr^2  + \frac{r-q}{D} dq^2 + D\frac{(d\tau-(r-\varepsilon)d\phi)^2}{(r-q)(q^2-\varepsilon^2)} \right)
\end{align}
The functions $C,D$ are then quartic polynomials, in which case the Einstein's equations are satisfied when 
\begin{align}
    C(r) &= -\varepsilon(4\Lambda + m \varepsilon^2 + n \varepsilon^3) + m \varepsilon^2 r + \frac{2n \varepsilon^3 + m \varepsilon^2 + 4\Lambda}{\varepsilon} r^2 - m r^3 - n r^4 \\
    D(q) &= -C(q)
\end{align}
where $m,n,\Lambda,\varepsilon$ are 4 parameters.
This gives another class of ambitoric Einstein metrics.

\subsection{Case 3}
For case 3 we do not need a change of coordinates and can remain in $x,y$. 
Similarly to the main text we require the trace of the Ricci tensor to be equal to the cosmological constant,
this becomes 
\begin{align}
    (1+\varepsilon (x+y))^2 (F_{xx} + G_{yy}) - 6 \varepsilon (1+\varepsilon (x+y))(F_x + G_y) + 12 \varepsilon^2 (F+G) = 24 (y-x) \Lambda \label{eq:Case2_Trace_Weyl_Condition}
\end{align}
taking derivatives w.r.t $x$ 3 times and alternatively with $y$ 3 times we obtain the equations 
\begin{align}
    F_{xxxxx} = 0, \quad G_{yyyyy} = 0
\end{align}
which have solutions such that $F,G$ are quartic polynomials in their respective variables.
Substituting this back into \cref{eq:Case2_Trace_Weyl_Condition} and solving for the constants reveals that 
\begin{align}
    F(x) &= b + m x + \gamma x^2 + (2\gamma \varepsilon - 2m \varepsilon^2 - 4 \Lambda) x^3 + \frac{\varepsilon(2\gamma \varepsilon - 2m \varepsilon^2 - 4\Lambda)}{2} x^4\\
    G(y) &= -F(y).
\end{align}
where $b,m,\gamma,\varepsilon,\Lambda$ are arbitrary constants.
This solution for $F,G$ automatically satisfies $R_{\mu\nu} = \frac{\Lambda}{4} g_{\mu \nu}$.
The Einstein metric is of the form
\begin{align}
    g = \frac{1}{(1+\varepsilon (x+y))^2} \left[ \frac{F}{x-y}(d\tau-y d\phi)^2 + \frac{x-y}{F} dx^2 + \frac{x-y}{G} dy^2 + \frac{G}{x-y} (d\tau- x d\phi)^2 \right].
\end{align}
\end{document}